\documentclass[preprint,3p,times]{elsarticle}

\usepackage[utf8]{inputenc} 
\usepackage[T1]{fontenc}    
\usepackage{hyperref}       
\usepackage{url}            
\usepackage{tabularx,booktabs}       
\usepackage{makecell}
\usepackage{nicefrac}       
\usepackage{microtype}      
\usepackage{xcolor}         
\usepackage{amsfonts,amsopn,amsthm,amsmath,amssymb}
\usepackage{graphicx,epstopdf}
\usepackage{subcaption}
\usepackage{algorithm,algpseudocode}
\usepackage{lineno}
\usepackage{multirow}
\usepackage{natbib}
\setcitestyle{numbers,square}

\newtheorem{theorem}{Theorem}[section]
\newtheorem{proposition}{Proposition}[section]
\newtheorem{definition}{Definition}[section]

\numberwithin{equation}{section}

\begin{document}

\begin{frontmatter}

\title{ALKPU: an active learning method for the DeePMD model with Kalman filter} 

\author[1]{Haibo Li}
\ead{haibo.li@unimelb.edu.au}

\author[2]{Xingxing Wu }
\ead{stars_sparkling@163.com}

\author[3,6]{Liping Liu}
\ead{liuliping@semi.ac.cn}

\author[3,6]{Lin-Wang Wang}
\ead{lwwang@semi.ac.cn}

\author[4]{Long Wang}
\ead{wanglong82@huawei.com}

\author[5,6]{Guangming Tan}
\ead{tgm@ict.ac.cn}

\author[5,6]{Weile Jia \corref{cor1}}
\ead{jiaweile@ict.ac.cn}

\affiliation[1]{organization={School of Mathematics and Statistics, The University of Melbourne},
	city={Melbourne, VIC},
	postcode={3010},
	country={Australia}}
\affiliation[2]{organization={Technology Department, Longxun Quantum Inc.},
	city={Beijing},
	postcode={100192},
	country={China}}
\affiliation[3]{organization={Institute of Semiconductors, Chinese Academy of Sciences},
	city={Beijing},
	postcode={100083},
	country={China}}
\affiliation[4]{organization={Computing System Optimization Lab, Huawei Technologies Co., Ltd.},
	city={Beijing},
	postcode={100094},
	country={China}}
\affiliation[5]{organization={State Key Lab of Processors, Institute of Computing Technology, Chinese Academy of Sciences},
	city={Beijing},
	postcode={100190},
	country={China}}
\affiliation[6]{organization={University of Chinese Academy of Sciences},
	city={Beijing},
	postcode={101408},
	country={China}}


\cortext[cor1]{Corresponding author}

\begin{abstract}
  Neural network force field models such as DeePMD have enabled highly efficient large-scale molecular dynamics simulations with \textit{ab initio} accuracy. However, building such models heavily depends on the training data obtained by costly electronic structure calculations, thereby it is crucial to carefully select and label the most representative configurations during model training to improve both extrapolation capability and training efficiency. To address this challenge, based on the Kalman filter theory we propose the \textit{Kalman Prediction Uncertainty (KPU)} to quantify uncertainty of the model's prediction. With KPU we design the \textit{Active Learning by KPU (ALKPU)} method, which can efficiently select representative configurations that should be labelled during model training. We prove that ALKPU locally leads to the fastest reduction of model's uncertainty, which reveals its rationality as a general active learning method. We test the ALKPU method using various physical system simulations and demonstrate that it can efficiently coverage the system's configuration space. Our work demonstrates the benefits of ALKPU as a novel active learning method, enhancing training efficiency and reducing computational resource demands.
\end{abstract}



\end{frontmatter}


\section{Introduction}\label{sec1}
The \textit{ab initio} molecular dynamics (AIMD) is a powerful computational tool that allows the simulation of dynamical processes involving atoms and molecules to gain a detailed understanding of physical and chemical properties of complex systems with quantum mechanics precision, such as chemical bonding and formation of reaction intermediates for metal alloys, ceramics and polymers \cite{nose1984,323_1985Computer,development1996,raty2005growth,aminpour2019}. However, AIMD is computationally expensive due to the need for electronic structure calculations to obtain the energy and forces of the system at each step, which limits the scale of the simulated system and time duration \cite{PhysRevLett.55.2471,marx_hutter_2009,pwmat1}. Recently, the deep learning methods provide tools to combine traditional AIMD with a neural network (NN) to develop neural network force field (NNFF) \cite{blank1995neural,2007Generalized,doi:10.1063/1.3553717,behler2021four}, in which an NN model is designed that maps configurations of the simulated system to their respective potential energies and forces. The well-trained NNFF model can be used to rapidly generate energies and forces to drive molecular dynamics (MD), improving the speed and scale of simulations while maintaining \textit{ab initio} accuracy. This approach has been effectively used for simulating complex condensed-phase systems as well as molecular transport studies \cite{Representing2004,2019Neural,friederich2021machine}.

 The performance of an NNFF depends heavily on the quality of training data generated by electronic structure calculation softwares such as VASP and PWmat \cite{hutchinson2012vasp,pwmat1,pwmat2}, and efficient sampling strategies play a crucial role in ensuring that the model is accurate and reliable across the entire range of configurations of interest \cite{2020exploring}. Thus, it is imperative to find such a method that selects informative and representative configurations to build a relative small training dataset with which we can train a model with the desired performance \cite{smith2018less,zhang2019active}. Active learning is a technique that concurrently selects the most informative data samples for labeling during training, which can significantly reduce the labeling effort and save much training time required to reach a given level of performance. 

Many NNFF models and packages have been developed, such as HDNNP~\cite{Behler_2014, doi:10.1063/1.3553717, behler2017first}, BIM-NN~\cite{doi:10.1021/acs.jpclett.7b01072}, Schnet~\cite{Schutt2017}, Enn-s2s~\cite{gilmer2017neural}, DeepMD-kit~\cite{wang2018deepmd}, SIMPLE-NN~\cite{lee2019simple}, DimeNet++~\cite{gasteiger_dimenet_2020, gasteiger_dimenetpp_2020}, PaiNN~\cite{schutt2021equivariant}, SpookyNet~\cite{Unke2021}, CabanaMD-NNP \cite{DESAI2022108156}, NewtonNet~\cite{haghighatlari2022newtonnet} and NequIP~\cite{Batzner2022}. The Deep Potential Molecular Dynamics (DeePMD) is one of the state-of-the-art model among them \cite{zhang2018deep,wang2018deepmd,zhang2018end,jia2020pushing}. It contains an embedding network to learning symmetric invariant features \cite{doi:10.1063/1.3553717,PhysRevB.87.184115} from input configurations and a fitting network to map the feature to potential energies and forces \cite{zhang2018end}. The model is trained using the Adam optimizer \cite{2014Adam} and an active learning method called DP-GEN is designed to select the most relevant portions of the high-dimensional configuration space \cite{zhang2020dp}. DP-GEN has two main parts: exploration and selection. Exploration involves using the temporary trained model to run several steps MD, while selection selects representative configurations from this MD trajectory that should be labeled to enlarge the training dataset. DP-GEN trains four models with different initial NN parameters, then uses the predictions of the trained model to calculate an empirical variance, and selects data points with high variances. Although it has been applied to many physical system simulations, DP-GEN has two main problems. Firstly, it incurs substantial computational expense due to the training of four models. Secondly, there is a consistency-error problem when using multiple models to calculate variance, which means that all four trained models may make consistently incorrect predictions for a given configuration \cite{zhang2020machine}, potentially leading to unnecessary redundant or missed acquisitions of configurations that the model has not covered.

In this paper, based on the RLEKF optimizer \cite{hu2022rlekf} that exploits the extended Kalman filter (EKF) \cite{kalman1960contributions,KF155276,KF728124,haykin2001kalman,chui2017kalman} to train DeePMD, we develop a new active learning method. By taking advantages of the Kalman filter's capabilities of estimating the probability distribution of NN parameters, we propose a metric to quantify uncertainty of the model's prediction called the Kalman Prediction Uncertainty (KPU), which is an approximate variance of the model's prediction. In the framework of the RLEKF optimizer, we propose an efficient algorithm to calculate KPU, then with KPU we design the Active Learning by KPU (ALKPU) method for DeePMD. After the exploration step, ALKPU computes the KPU for each configuration in the new MD trajectory, and configurations with higher KPUs are selected for labeling and then added to the training dataset concurrently during model training. Comparing with DP-GEN's four-model method, ALKPU is more appropriate to quantify uncertainty of the model's prediction and trains only one model, thus can save large computational resource overheads while lead to a better selection of representative data points. The main contributions of this work are summarized in the following.

\subsection{Contributions}
\begin{itemize}
  \item Based on the EKF for training NN, we propose the KPU to quantify uncertainty of the NN's predictions. In the framework of the RLEKF optimizer for training DeePMD, we design an efficient algorithm to calculate the KPU of forces predicted by the temporary trained model. 
  \item We propose the ALKPU active learning method for DeePMD and develop a sampling platform that can concurrently exploring, selecting and labeling representative configurations during training. Several typical physical system are used to test the ALKPU method and show its advantages as a new active learning method.
  \item We establish some theoretical results about ALKPU with mathematical proofs by using the Fisher information matrix and Shannon entropy. The results show that ALKPU is essentially a locally fastest reduction procedure of NN model's uncertainty, which reveals the rationality of ALKPU as a general active learning method. 
\end{itemize}

\subsection{Related work}
\paragraph*{Active learning for neural networks}
Active learning is a widely studied topic, with early work found in \cite{2010Active}. The center of active learning is the query strategy that decides which data points are informative and representative that should be selected and labelled \cite{angluin2001queries}. DP-GEN belongs to the ``query by committee'' family where a variety of models are trained on the current labeled data and then votes on the output for unlabeled data and labels those points for which the ``committee'' disagrees the most \cite{seung1992query}. ALKPU belongs to the ``uncertainty sampling'' family that labels those points for which the current model is the least certain about regarding the correct output \cite{Faria2022}. Other query strategies include the variance reduction that labels points that would minimize output variance \cite{mackay1992information,cohn1993neural,cohn1996active}, the expected error reduction that label points that would mostly reduce the model's generalization error \cite{roy2001toward} and the expected model change that labels points that would mostly change the current model \cite{settles2007multiple}. Statistical learning methods such as the Gaussian regression process and Bayesian neural network can give better uncertainty measures \cite{rasmussen2006gaussian,deringer2021gaussian}, but it is hard to directly applied them to deep neural networks. Recent research \cite{gal2016dropout,li2017dropout,gal2017concrete} has demonstrated connections between the dropout based active learning and approximate Bayesian inference, making it possible to apply Bayesian active learning methods to deep learning.
 
\paragraph*{Active learning in MD simulations}
Active learning is particularly useful for data-extensive scenarios such as chemical-space exploration in MD simulations. One effective algorithm proposed in \cite{botu2015adaptive} to actively learn energies and forces for MD of bulk aluminium and aluminium surfaces employs Euclidean distance between configurations as a selection criterion. The ``query by committee'' was applied first for NNFFs in \cite{behler2015constructing} where a committee of NN models with different architectures but comparable average performances is selected. In contrast to NNs, in \cite{simm2018error} the prediction variance of Gaussian regression process was used to learn complex chemical reaction mechanisms. This approach was further extended to enable the generation of multiple labels simultaneously, providing a significant acceleration of matter simulations \cite{proppe2019gaussian}. For the linearly parametrized interatomic potential \cite{podryabinkin2017active}, an active learning approach using the D-optimality criterion for selecting configurations was designed. Additionally, active learning approaches have been applied to accelerate chemical-space exploration based on the Thompson sampling \cite{hernandez2017parallel} and to speed up exploration of potential energy surfaces by synthesizing representations of unseen configurations \cite{gomez2018automatic}.

\section{The DeePMD model and RLEKF optimizer}\label{sec2}
The state-of-the-art NNFF used in this work is the DeePMD model \cite{zhang2018deep,zhang2018end}. The two main parts of DeePMD are the embedding net and fitting net. For each atom $i$ in a given configuration, the physical symmetries such as translational, rotational, and permutational invariance are integrated into the descriptor $\mathcal{D}_{i}$ through a fully connected network named embedding net. Then the descriptor passes through a fully connected network named fitting net to approximate local atomic energy. The network structure can be described as follows.

\textbf{(1).} For each configuration of the given system consisting Cartesian coordinates $\mathbf{r}_{i}=(x_{i},y_{i},z_{i})\in \mathbb{R}^{3}$ for $i=1,\dots,N_{a}$, the input of DeePMD is the neighbor lists $\mathcal{R}_{i}=\{\mathbf{r}_{ij}=\mathbf{r}_{j}-\mathbf{r}_{i}:|\mathbf{r}_{ij}| <r_{c}\}$ for all $i$, where $r_c$ is the cutoff radius. Then $\mathcal{R}_{i}$ is translated into the smooth version $\tilde{\mathcal{R}}_{i} \in \mathbb{R}^{N_{m}\times4}$ with $(\tilde{\mathcal{R}}_{i})_{j}=s(|\mathbf{r}_{ij}|)(1,\mathbf{r}_{ij}/|\mathbf{r}_{ij}|)\in \mathbb{R}^4$, where $N_{m}$ is the maximum length of all neighbor lists and $j$ is the neighbor index of atom $i$. The smoothing function $s(r)$ satisfies $s(r)=1/r$ for $x<r_{cs}$ and $s(r)=0$ for $x>r_{c}$, and it decaying smoothly between the two thresholds.

\textbf{(2).} Define the mapping of the embedding network as $\mathcal{G}=\mathcal{E}_{2}\circ\mathcal{E}_{1}\circ\mathcal{E}_{0}$, where $\mathcal{E}_{0}(\boldsymbol{x})=\tanh \left(\mathbf{\boldsymbol{x}}\otimes W_{0}+\boldsymbol{1}\otimes w_{0}\right)$ with $W_0\in \mathbb{R}^{M_1}$ and $\mathcal{E}_{l}(X)=X+\tanh \left(XW_{l}+\boldsymbol{1}\otimes w_{l}\right)$ with $W_{l}\in \mathbb{R}^{N_{m} \times M_1}$ for $l=1,2$. The notation ``$\otimes$'' is the cross product between vectors $\boldsymbol{1}\in \mathbb{R}^{N_{m}}$ and $w_{l} \in \mathbb{R}^{M_1}$, and $\tanh$ is element-wise. Then we define the local embedding matrix $\mathcal{G}_{i}\in \mathbb{R}^{N_{m}\times M_1}$ with its elements $(\mathcal{G}_{i})_{jk}=\mathcal{G}(s(|\mathbf{r}_{ij}|))_k$. The descriptor of atom $i$ is defined as $\mathcal{D}_{i}:=\mathcal{G}_{i1}^{T}\tilde{\mathcal{R}}_{i}\tilde{\mathcal{R}}_{i}^{T}\mathcal{G}_{i2}\in \mathbb{R}^{M_1\times M_2}$. In practice, we take $\mathcal{G}_{i1}=\mathcal{G}_{i}$ and take the first $M_2$ ($<M_1$) columns of $\mathcal{G}_i$ to form $\mathcal{G}_{i2}$.

\textbf{(3).} The mapping of the fitting network is defined as $\mathcal{F}(\mathcal{D}_{i})=\mathcal{F}_{3}\circ\mathcal{F}_{2}\circ\mathcal{F}_{1}\circ\mathcal{F}_{0}(\mathcal{D}_{i})$ with $\mathcal{D}_{i}$ reshaped into a vector of form $\mathbb{R}^{M_1 M_2}$, where $\mathcal{F}_{0}(\boldsymbol{x})=\tanh \left(\tilde{W}_{0}\boldsymbol{x}+\tilde{w}_{0}\right)$ with $\tilde{w}_{0} \in \mathbb{R}^{d}$, $\tilde{W}_{0}\in \mathbb{R}^{d \times M_1 M_2}$, $\mathcal{F}_{3}(\boldsymbol{x})=\tilde{W}_{3}\boldsymbol{x}+\tilde{w}_{3}$ with $\tilde{w}_{3} \in \mathbb{R}$, $\tilde{W}_{3}\in \mathbb{R}^{1 \times d}$ and $\mathcal{F}_{l}(\boldsymbol{x})=\boldsymbol{x}+\tanh \left(\tilde{W}_{l}\boldsymbol{x}+\tilde{w}_{l}\right)$ with $\tilde{w}_{l} \in \mathbb{R}^{d}, \tilde{W}_{l}\in \mathbb{R}^{d \times d}$ for $l=1,2$.

\textbf{(4).} The output $E_{i}=\mathcal{F}(\mathcal{D}_{i})$ is the local energy corresponding to atom $i$. Finally the total energy of this system is $E=\sum_{i}E_{i}$ and the force for each atom is $\boldsymbol{F}_{i}=-\nabla_{\mathbf{r}_{i}}E$ that can be computed by backward propagation.

The DeePMD-kit software \cite{wang2018deepmd} exploits the Adam optimizer \cite{2014Adam} to train the NN parameters with the loss function be a weighted sum of energy loss and force loss. In this work, the model is trained by the RLEKF optimizer proposed in \cite{hu2022rlekf}, which is the abbreviation of the \textit{reorganized layer extended Kalman filter}. The idea of exploiting EKF to train NN is as follows. Let $\boldsymbol{w}$ denotes the vector including all the well-trained parameters. The NN model can be modelled as 
\begin{equation*}
  \begin{cases}
    \boldsymbol{v}_{t}=\boldsymbol{v}_{t-1}=\boldsymbol{w}, \\
    y_{t}=h(\boldsymbol{v}_{t},x_{t})+\eta_{t},
  \end{cases}
\end{equation*}
where $\{(x_{t},y_{t})\}_{t\in \mathbb{N}}$ are training data pairs, $h(\cdot, \cdot)$ is the nonlinear NN mapping, and $\eta_{t}$ is the noise for the NN model. In order to avoid overfitting of earlier data, a fading memory strategy is used, which leads to the following model:
\begin{equation}\label{fade}
  \begin{cases}
    \boldsymbol{v}_{t}=\lambda_{t}^{-1/2}\boldsymbol{v}_{t-1}, \ \ \
    \boldsymbol{v}_{0} \sim \mathcal{N}(\hat{\boldsymbol{w}}_0, P_0), \\
    y_{t}=h(\alpha_{t}^{-1}\boldsymbol{v}_{t},x_{t})+\eta_{t},  \ \ \  \eta_{t}\sim{\mathcal {N}} (0, R_{t}),
  \end{cases}
\end{equation}
where $0< \lambda_{i}\leq 1$ are called fading memory factors, $\boldsymbol{v}_{0}=\boldsymbol{w}$ and $\mathcal{N}(\hat{\boldsymbol{w}}_0, P_0)$ is the initial distribution of $\boldsymbol{w}$. Using the relation $\boldsymbol{v}_t=\alpha_{t}\boldsymbol{w}$ with $\alpha_t=\lambda_{t}^{-1/2}\cdots\lambda_{1}^{-1/2}$, the $t$-th trained parameter $\hat{\boldsymbol{w}}_t$ is defined as $\hat{\boldsymbol{w}}_t=\alpha_{t}^{-1}\hat{\boldsymbol{v}}_{t}$ with $\hat{\boldsymbol{v}}_{t}$ computed by EKF. The procedure is summarized in Algorithm \ref{EKFfm}.

\begin{algorithm}
  \caption{EKF with fading memory (EKFfm) for training NN}\label{EKFfm}
  \begin{algorithmic}[1]
    \Require $\hat{\boldsymbol{w}}_0$, $P_0$, $\lambda_1$, $\nu$
    \For{$t=1,2,\ldots,T$}
    \State $\bar{P}_{t} = P_{t-1}\lambda_{t}^{-1}$
    \State $H_t = \mathrm{D}_{\boldsymbol{w}}h(\hat{\boldsymbol{w}}_{t-1},x_t)$  \Comment{Differential (Jacobian) with respect to $\boldsymbol{w}$}
    \State $A_t = H_{t}\bar{P}_{t}H_{t}^T+\alpha_{t}^{2}R_{t}$, \ \
    $\alpha_t = \lambda_{t}^{-1/2}\cdots\lambda_{1}^{-1/2}$
    \State $K_t = \bar{P}_{t}H_{t}^{T}A_{t}^{-1}$
    \State $P_{t} = (I-K_{t}H_{t})\bar{P}_{t}$ \Comment{Covariance matrix of estimated $\hat{\boldsymbol{v}}_{t}$}
    \State $d_t = y_t - h(\hat{\boldsymbol{w}}_{t-1}, x_t)$
    \State $\hat{\boldsymbol{w}}_{t} = \hat{\boldsymbol{w}}_{t-1}+K_{t}d_t$ \Comment{Trained parameter}
    \State $\lambda_{t+1} = \lambda_{t} \nu +1- \nu $ \Comment{$0<\nu<1$}
    \EndFor
    \Ensure $\hat{\boldsymbol{w}}_T$, $P_T$, $\lambda_{T+1}$
  \end{algorithmic}
\end{algorithm}

The EKFfm is further optimized by elaborately reorganizing the NN layers by splitting big and gathering adjacent small ones to approximate a dense matrix by a sparse diagonal block matrix, and the energy and force are trained alternatively where the prediction for force is transformed to a scalar, which makes $A_t$ a scalar to avoid matrix inversions. After code optimizations, RLEKF converges faster with slightly better accuracy for both force and energy than Adam \cite{hu2022rlekf}. The detailed derivation of EKFfm and descriptions of RLEKF are described in \ref{sec:A}.

\section{Active learning by Kalman Prediction Uncertainty}\label{sec3}
Active learning minimizes labeling costs while maximizing modeling accuracy. While there are various methods in active learning literature, we propose to label the data point whose model uncertainty is the highest. Thanks to the EKF that estimates the distribution of trained NN parameter, we can estimate the variance around the model's prediction and use it as an uncertainty measure.

\subsection{Kalman Prediction Uncertainty}
After $t$ iterations of training by EKFfm, we have estimated parameter $\boldsymbol{w}$ with expectation $\hat{\boldsymbol{w}}_t$ and covariance $\hat{P}_t=\alpha_{t}^{-2}P_t$ due to $\hat{\boldsymbol{w}}_t=\alpha_{t}^{-1}\hat{\boldsymbol{v}}_{t}$. If we treat the $t$-th trained $\boldsymbol{w}$ as Gaussian $\boldsymbol{w}\sim\mathcal{N}(\hat{\boldsymbol{w}}_t,P_t)$ and approximate $h(\boldsymbol{w},x)$ near $\hat{\boldsymbol{w}}_t$ by the one-order Taylor expansion $h(\boldsymbol{w},x)\approx h(\hat{\boldsymbol{w}}_t,x)+D_{\boldsymbol{w}}h(\hat{\boldsymbol{w}}_t,x)(\boldsymbol{w}-\hat{\boldsymbol{w}}_t)$, then $h(\boldsymbol{w},x)$ approximately obeys the distribution 
\[\mathcal{N}(h(\hat{\boldsymbol{w}}_t,x), \mathrm{D}_{\boldsymbol{w}}h(\hat{\boldsymbol{w}}_t,x)\hat{P}_{t}\mathrm{D}_{\boldsymbol{w}}h(\hat{\boldsymbol{w}}_t,x)^T)=
\mathcal{N}(h(\hat{\boldsymbol{w}}_t,x), \alpha_{t}^{-2}\mathrm{D}_{\boldsymbol{w}}h(\hat{\boldsymbol{w}}_t,x)P_{t}\mathrm{D}_{\boldsymbol{w}}h(\hat{\boldsymbol{w}}_t,x)^T) .\]
By neglecting the scalar factor $\alpha_{t}^{-2}$ at each iteration, we define the \textit{Kalman Prediction Uncertainty (KPU)} as follows to quantity uncertainty of the model's prediction for an input $x$.

\begin{definition}[KPU]
  For an NN model trained by EKFfm, define the KPU of an input $x$ with respect to the $t$-th trained parameters as
  \begin{equation}\label{KPU}
    \mathrm{KPU}_t(x) := \mathrm{D}_{\boldsymbol{w}}h(\hat{\boldsymbol{w}}_t,x)P_{t}\mathrm{D}_{\boldsymbol{w}}h(\hat{\boldsymbol{w}}_t,x)^T .
  \end{equation}
\end{definition}

When the prediction has dimension 1, the KPU becomes $\mathrm{KPU}_t(x) = \nabla_{\boldsymbol{w}}h(\hat{\boldsymbol{w}}_t,x)^{T}P_{t}\nabla_{\boldsymbol{w}}h(\hat{\boldsymbol{w}}_t,x)$ that is a nonnegative number. Note that $\nabla_{\boldsymbol{w}}h(\hat{\boldsymbol{w}}_t,x)=\mathrm{D}_{\boldsymbol{w}}h(\hat{\boldsymbol{w}}_t,x)^{T}$. For the RLEKF optimizer, the computed $P_t$ has a block diagonal form corresponding to the reorganized layer strategy, thus KPU can be calculated as Algorithm \ref{KPU_layer}. Since RLEKF is a block diagonal approximation to EKFfm, the computed KPU in RLEKF is an approximation to the real KPU defined in \eqref{KPU}. 

\begin{algorithm}[!htbp]
  \caption{KPU calculation in RLEKF ($\mathtt{KPU\_cal}$)}\label{KPU_layer}
  \begin{algorithmic}[1]
    \Require $\hat{\boldsymbol{w}}_{t}$, $P_t=\mathrm{diag}(P_{t1},\dots,P_{tL})$, input $x$ \Comment{$P_t$ is block diagonal}
    \State $H = \mathrm{D}_{\boldsymbol{w}}h(\hat{\boldsymbol{w}}_{t}, x)$ \Comment{Backward propagation}
    \State Split $H$ corresponding to $P_t$: $\{H_1, \dots, H_L\}=\mathtt{split}(H)$ 
    \For{$l =1,\dots, L$} \Comment{$L$ is the number of reorganized layers}
      \State $ KPU_l = H_{l}P_{tl}H_{l}^{T} $\Comment{KPU for each reorganized layer}
  \EndFor
  \Ensure $\mathrm{KPU}_t(x)=\sum_{l=1}^{L}KPU_l$
  \end{algorithmic}
\end{algorithm}

\begin{algorithm}[htbp]
  \caption{KPU for atomic force ($\mathtt{KPU_{force}}$)}\label{KPU_cal}
  \begin{algorithmic}[1]
    \Require $\hat{\boldsymbol{w}}_t$, $P_t$, configuration $\mathcal{R}=\{\mathcal{R}_{i}\}_{i=1}^{N_a}$
    \State $ E = \sum_{i} h_E(\hat{\boldsymbol{w}}_{t}, \mathcal{R}_{i})$ \Comment{Predicted total energy}
    \State $\boldsymbol{F}_{i}=-\nabla_{\mathbf{r}_{i}}E$ \Comment{Predicted atomic force}
    \State $\{\boldsymbol{F}_1,\dots,\boldsymbol{F}_n\}=\mathtt{rand\_select}(N_a,\mathtt{frac})$\Comment{Randomly select $n=\lfloor N_a\cdot \mathtt{frac}\rfloor$ atoms}
      \For{$j =1,\dots,n$}
          \For{$k = 1, 2, 3$} 
          \State $KPU_{F_{jk}}=\mathtt{KPU\_cal}(\hat{\boldsymbol{w}}_{t}, P_{t}, \mathcal{R}; F_{jk})$ \Comment{KPU for force component $F_{jk}$}
          \EndFor
      \EndFor
    \Ensure $\mathrm{KPU}_{\mathrm{force}}= \sum_{j=1}^{n} \sum_{k=1}^{3}KPU_{F_{jk}}/3n$ \Comment{Average of KPU for forces}
  \end{algorithmic}
\end{algorithm}

Although KPU for both total energy and atomic forces can be calculated, we find that using KPU for force is generally better in the practical active learning procedure. This is because forces represent local information of atoms which are more sensitive to the changes of configurations, while energy is a global quantity and does not seem to provide sufficient information. Since a physical system involves many atoms, to reduce computational cost and make KPU to be a scalar, a subset of atoms is randomly selected to calculate and sum the KPU for the three components of these atomic forces. Denote by $h_E(\cdot,\cdot)$ the nonlinear mapping from $\mathcal{R}_{i}$ to the atomic energy $E_i$, the calculation of KPU for force is described in Algorithm \ref{KPU_cal}. The output $\mathrm{KPU}_{\mathrm{force}}$ is a nonnegative number that can measure uncertainty for the predicted forces.

\subsection{ALKPU for DeePMD}
During model training, the active learning by KPU (ALKPU) method concurrently selects the next data points with higher KPUs. This gives the following procedure of ALKPU for DeePMD:
\begin{enumerate}
  \item Exploration. With the initial training dataset, a temporary model is trained by RLEKF. Then several steps of MD is run using the energy and forces generated by this NN model. This procedure is used to explore the high-dimensional configuration space to obtain some configurations not too far away from the real AIMD trajectory.
  \item Selection. For the configurations obtained by the exploration procedure, we need to select the most representative ones. The $\mathtt{KPU_{force}}$ of each configuration is computed and then judged whether $\sigma_0\leq \mathtt{KPU_{force}}\leq \sigma_1$ with $\sigma_0=c_0\mathrm{KPU}_{\mathrm{trained}}$ and $\sigma_1=c_1\mathrm{KPU}_{\mathrm{trained}}$, where $\mathrm{KPU}_{\mathrm{trained}}$ is the average KPU of some training data at the last training step. The threshold $\sigma_0$ is used to select configurations with higher KPUs that needed to be labelled. If the KPU is higher than $\sigma_1$, it means that this configuration deviates too far from the real AIMD trajectory and thus it should be discarded.
  \item Labelling. For those selected configurations, labelling can be done by using an electronic structure calculation software to compute the corresponding total energy and atomic forces.
  \item Training. The above labelled data are added to the previous training dataset. With the enlarged training dataset, the model is trained again. This cycle can be repeated multiple times until a well-trained model with desired performance is obtained. 
\end{enumerate}
The workflow of the whole procedure is shown in Figure \ref{fig1}.

\begin{figure}[htb] 
  \centering \includegraphics
  [width=150mm]{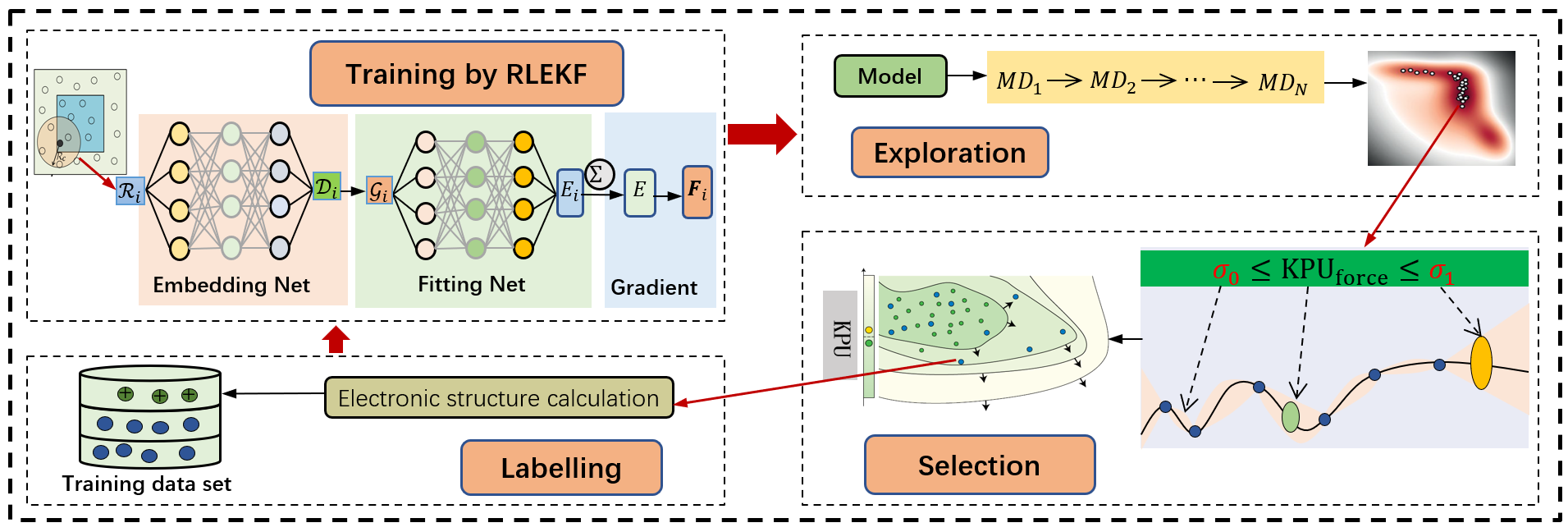}
  \caption{Workflow of ALKPU for DeePMD}
  \label{fig1}
\end{figure}

\section{Theoretical analysis of ALKPU}\label{sec4}
The ALKPU method is essentially a locally fastest reduction procedure of model's uncertainty. This can be revealed with the help of the Fisher information matrix (FIM) and Shannon entropy. The proofs of all the results are in \ref{sec:C}. Suppose the training data $\{(x_{t},y_{t})\}_{t\in \mathbb{N}}$ satisfies the parametrized statistical model 
\begin{equation}\label{data_model}
  y_t = h(\boldsymbol{w},x_t) + \eta_t, \ \ \ \eta_t \sim \mathcal{N}(0,R_t).
\end{equation}
The following result establishes a correspondence between $P_t$ and FIM. Note that $v^{\otimes 2}:=vv^T$ is the cross product for a column vector $v$.

\begin{theorem}\label{ONG}
  For the EKFfm with training data $\{(x_{t},y_{t})\}_{t\in \mathbb{N}}$, denote by $p(y_t|\boldsymbol{w})$ the probability density function and let $J_0 = \gamma_{0}P_{0}^{-1}$ for any arbitrary $\gamma_0 >0$. Then at each step, the following relation holds:
  \begin{equation}\label{Jt}
    J_t = \beta_{t}J_{t-1}+(1-\beta_t)\mathbb{E}_{y_t\sim p(y_t|\hat{\boldsymbol{w}}_{t-1})}\Big [\nabla_{\boldsymbol{w}}\ln p(y_t|\hat{\boldsymbol{w}}_{t-1})^{\otimes 2}\Big ], \ \ \ J_t=\gamma_t P_{t}^{-1},
  \end{equation}
  where
  \begin{equation}\label{beta}
    \beta_t = \frac{\lambda_t}{\lambda_{t}+\alpha_{t}^{-2}\gamma_{t-1}} ,
  \end{equation}
  and $\gamma_t$ satisfies the recursion
  \begin{align}\label{gamma}
    \gamma_t = \frac{\gamma_{t-1}}{\lambda_t + \alpha_{t}^{-2}\gamma_{t-1}}, \ \ \ t\geq 1 .
  \end{align}
\end{theorem}

\begin{proposition}\label{limit}
  For $0<\nu <1$, the asymptotic behavior of $\gamma_t$ and $\beta_t$ satisfies
  \begin{align}
    \lim_{t\rightarrow \infty}\gamma_t = 0, \ \ \ \lim_{t\rightarrow \infty}\beta_t = 1 .  \label{lim}
  \end{align}
\end{proposition}

For model \eqref{data_model}, the FIM is symmetric semi-definite matrix defined as
\begin{equation}\label{FIM}
  J(\boldsymbol{w}):=\mathbb{E}_{y_t\sim p(y_t|\boldsymbol{w})}\Big[\nabla_{\boldsymbol{w}}\ln p(y_t|\boldsymbol{w})^{\otimes 2}\Big]
  = -\mathbb{E}_{y_t\sim p(y_t|\boldsymbol{w})}\Big[\nabla_{\boldsymbol{w}}^{2}\ln p(y_t|\boldsymbol{w})\Big] ,
\end{equation}  
where $\nabla_{\boldsymbol{w}}^{2}$ denotes the Hessian with respect to $\boldsymbol{w}$ \cite{amari2016}. Theorem \ref{ONG} implies that $J_t$ is a weighted sum of the FIMs with respect to every estimated parameter $\hat{\boldsymbol{w}}_{i}$ during training steps $i=1,\dots,t$. The FIM can be used to measure the amount of information that the training data carry about the unknown parameter $\boldsymbol{w}$, since the inverse of FIM sets a lower bound on the variance of the model's parameter estimate, which is known as the Cram{\'e}r-Rao inequality \cite{Rao1945,Cramer1946}. Relation \eqref{Jt} implies that a smaller $J_{t}^{-1}$ is equivalent to a smaller $P_t$, meaning a lower variance of the $t$-th estimated $\boldsymbol{w}$. Note that for a multidimensional parameter $\boldsymbol{w}$, one has to map the FIM to a scalar for comparison.

Suppose we have the $t$-th trained model, in order to further increase model's extrapolation ability, the next selected data should be less similar as which the $t$-th model can predict, thereby an active learner should select a new training data that leads to a big reduction of the model's uncertainty. Suppose we have a prior distribution of $\boldsymbol{w}$ and the data $(x, y)$ obeys the parametrized model
\begin{equation}\label{model2}
  y=h(\boldsymbol{w},x)+\eta .
\end{equation}
After receiving the data $(x, y)$, the posterior distribution of $\boldsymbol{w}$ becomes $\mathbb{P}(\boldsymbol{w}|y)= \mathbb{P}(y|\boldsymbol{w})\mathbb{P}(\boldsymbol{w})/\mathbb{P}(y)$ by the Bayes' formula. The reduction of uncertainty of $\boldsymbol{w}$ can be measured by the change of Shannon entropy defined as follows.

\begin{definition}[Change of entropy]
  For the data $(x, y)$ obeys the model \eqref{model2}, suppose the probability density functions of $\boldsymbol{w}$ before and after receiving data $(x,y)$ are $p(\boldsymbol{w})$ and $p(\boldsymbol{w}|y)$. Define the change of entropy of $\boldsymbol{w}$ as
  \begin{equation}
    S := \mathbb{E}_{\boldsymbol{w}\sim p(\boldsymbol{w}|y)}[\ln p(\boldsymbol{w}|y)] - 
    \mathbb{E}_{\boldsymbol{w}\sim p(\boldsymbol{w})}[\ln p(\boldsymbol{w})] ,
  \end{equation}
  which is the difference of Shannon entropy between $p(\boldsymbol{w})$ and $p(\boldsymbol{w}|y)$.
\end{definition}

\begin{theorem}\label{kpu_gain}
  Suppose the data $(x, y)$ obeys the model \eqref{model2} with $\eta\sim\mathcal{N}(0,\kappa^{-1}I)$ and $\boldsymbol{w}\sim\mathcal{N}(\hat{\boldsymbol{w}}_t, \hat{P}_t)$ at the $t$-th EKFfm iteration. If we use the approximation
  \begin{equation}\label{approx}
    \frac{1}{2}\nabla^{2}_{\boldsymbol{w}}\|y-h(\boldsymbol{w}, x)\|^2 \approx \nabla_{\boldsymbol{w}}h(\boldsymbol{w},x)\nabla_{\boldsymbol{w}}h(\boldsymbol{w},x)^T
  \end{equation}
  by omitting $\partial^2/\partial  w_i\partial  w_j$ and higher order derivative terms, then the change of entropy of $\boldsymbol{w}$ at the $t$-th iteration is
  \begin{equation}
    S_t = \frac{1}{2}\ln\left(\det\left(I + \kappa\alpha_{t}^{-2}\mathrm{KPU}_{t}(x) \right) \right).
  \end{equation}
\end{theorem}

A similar result is proposed in \cite{mackay1992information}, where the approximation \eqref{approx} is also used if $(x, y)$ falls into a flat region around $\boldsymbol{w}$. Note that $S_t = \frac{1}{2}\ln\left(1+\kappa\alpha_{t}^{-2}\mathrm{KPU}_t(x)\right)$ when the prediction $y$ has dimension 1. Theorem \ref{kpu_gain} indicates that ALKPU selects the new training data that maximizes the reduction of entropy from the prior to the posterior, thus leads to a locally fastest reduction of model's uncertainty. 

For a well-trained model that can represent the data, there is a relation between the KPU and expected loss that corresponds to mean square error (MSE). For simplicity, the result presented below assumes that $y$ has dimension 1.

\begin{theorem}\label{kpu_loss}
  Suppose data $(x, y)$ obeys the model \eqref{model2} with $\boldsymbol{w}\sim\mathcal{N}(\hat{\boldsymbol{w}}_t,\hat{P}_t)$. Let $\hat{y}_t=h(\hat{\boldsymbol{w}}_t,x)$. If we use the approximation $h(\boldsymbol{w},x)\approx h(\hat{\boldsymbol{w}}_t,x)+\nabla_{\boldsymbol{w}}h(\hat{\boldsymbol{w}}_t,x)^T(\boldsymbol{w}-\hat{\boldsymbol{w}}_t)$, then
  \begin{equation}
    \mathbb{E}_{\boldsymbol{w}}[|y-\hat{y}_t|^2] \geq \alpha_{t}^{-2}\mathrm{KPU}_{t}(x) .
  \end{equation}
\end{theorem}

The linear approximation of $h(\boldsymbol{w},x)$ around $\hat{\boldsymbol{w}}_t$ is reasonable given that the model's prediction is fairly good. Theorem \ref{kpu_loss} implies that for the data well covered by the trained model, KPU can be used to predict the MSE loss.

\section{Experimental results}\label{sec5}
In the experiments, the network configuration is $[1, 25, 25, 25]$ for the embedding net and  $[400, 50, 50, 50,1]$ for the fitting net, and the setting of $r_c$, $N_m$, $M_1$, $M_2$ and $s(\cdot)$ is consistent with DeePMD-kit \cite{wang2018deepmd}. For RLEKF, we set $\lambda_{1}=0.98$, $\nu=0.9987$ and $P_{0}=I$, while the initialization $\hat{\boldsymbol{w}}_0$ is consistent with DeePMD-kit and the reorganizing-layer strategy is consistent with \cite{hu2022rlekf}. For the KPU calculation procedure $\mathtt{KPU}_{\mathrm{force}}$, we set $\mathtt{frac}=0.5$. The two threshold $c_0$ and $c_1$ in ALKPU are set as $c_0=1.25$ and $c_1=2.0$. The AIMD and electronic structure calculations based on density functional theory (DFT) are implemented using the PWmat software \cite{pwmat1,pwmat2}.

The experimental environment is a high-performance cluster environment, including six computing nodes and one proxy node; see Table \ref{tabD1} for a detailed description. Each computing node has four high-performance graphics cards, including four nodes partitioned as ``A100'', each with four high-performance A100 graphics cards where the memory of a single card is 40GB. Two other nodes are partitioned as ``3090'', each with four GeForce RTX 3090 graphics cards, where the memory of a single card is 24GB.
In the experiments of ALKPU and DP-GEN, the main process runs on the proxy node, responsible for task distribution, result handling, and driving the iterative active learning process. Model training and DFT calculation were run on the computing nodes.

\begin{table}[H]
  \caption{Description of the high performance cluster environment}\label{tabD1}
    \centering
    \resizebox{\linewidth}{!}{
    \begin{tabular}{|c|c|c|c|c|}
      \hline
    \multirow{2}{*}{\makecell[l]{Computing \\ environment}} & \multirow{2}{*}{Proxy node} & \multicolumn{2}{c|}{Partition ``A100''} & Partition ``3090'' \\
          \cline{3-5} 
     & & Node 1, 2, 3 & Node 4 & Node 5, 6 \\
            \hline
    CPU & \makecell[l]{ Intel(R) Xeon(R) \\ Silver 4210 \\ CPU @ 2.20GHz} & 
          \makecell[l]{ Intel(R) Xeon(R) \\ Platinum 8375C \\ CPU @ 2.90GHz} &   
          \makecell[l]{ Intel(R) Xeon(R) \\ Gold 6248R \\ CPU @ 3.00GHz}  &
          \makecell[l]{ Intel(R) Xeon(R) \\ Gold 6326 \\ CPU @ 2.90GHz}  \\ \hline
    Number of CPU  & 2  & 2          & 2        & 2  \\ \hline
    CPU cores    & 10  & 32           & 24       & 16                      \\ \hline
    Logistical CPU  & 40  & 64        & 48       & 32                       \\ \hline
    Graphics card  & no & \multicolumn{2}{c|}{GRID A100 PCIe 40GB * 4} & GeForce RTX 3090 24G * 4 \\ 
    \hline
     \end{tabular}}
\end{table}

\begin{figure}[htbp]
  \centering
  \subcaptionbox{ \label{D2a} Cu}{\includegraphics[width=0.25\textwidth]{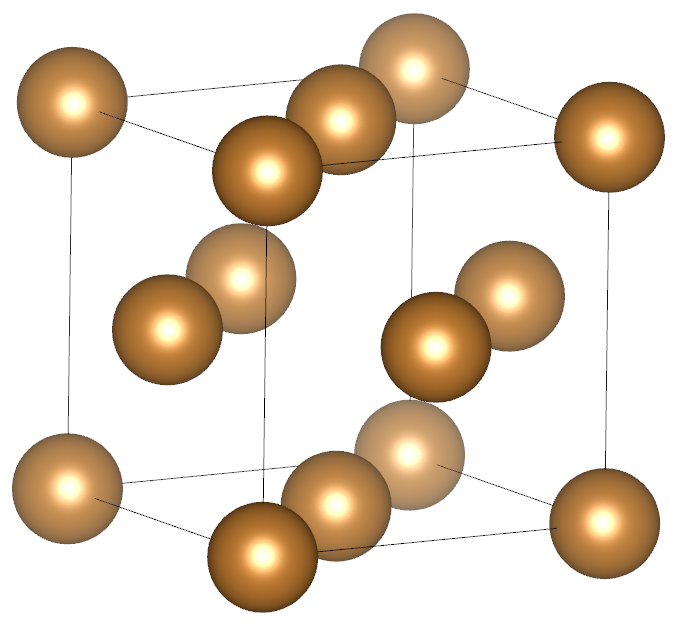}}\hspace{10mm}
  \subcaptionbox{ \label{D2b} Si}{\includegraphics[width=0.28\textwidth]{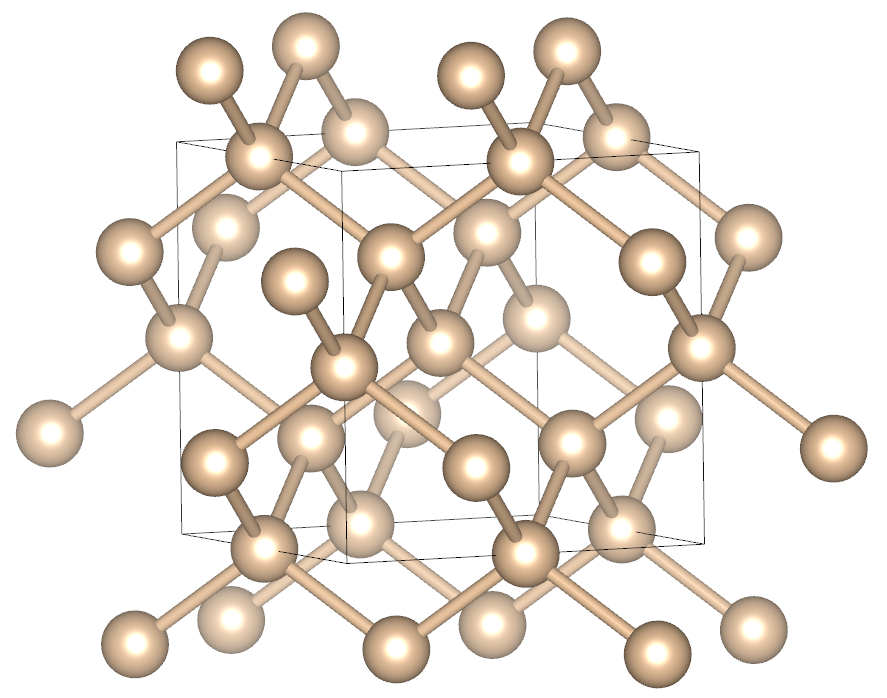}}%
  \vfill
  \subcaptionbox{ \label{D2c} Al}{\includegraphics[width=0.25\textwidth]{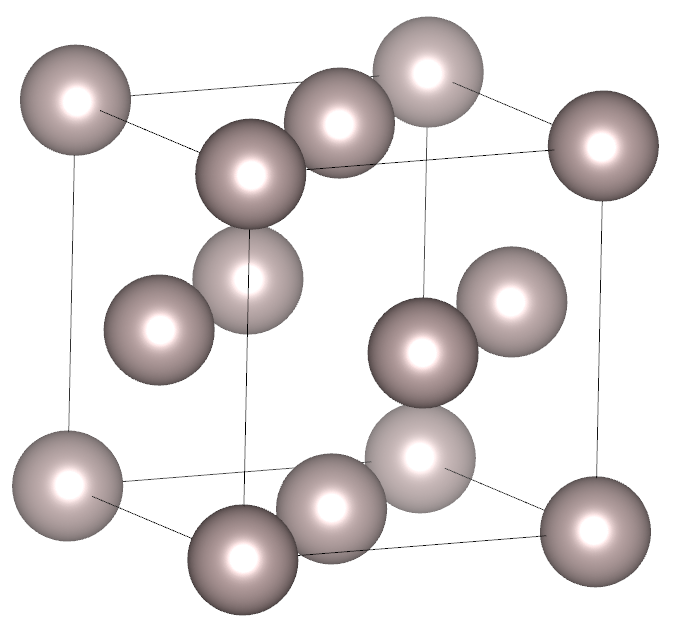}}\hspace{10mm}
  \subcaptionbox{ \label{D2d} Ni}{\includegraphics[width=0.25\textwidth]{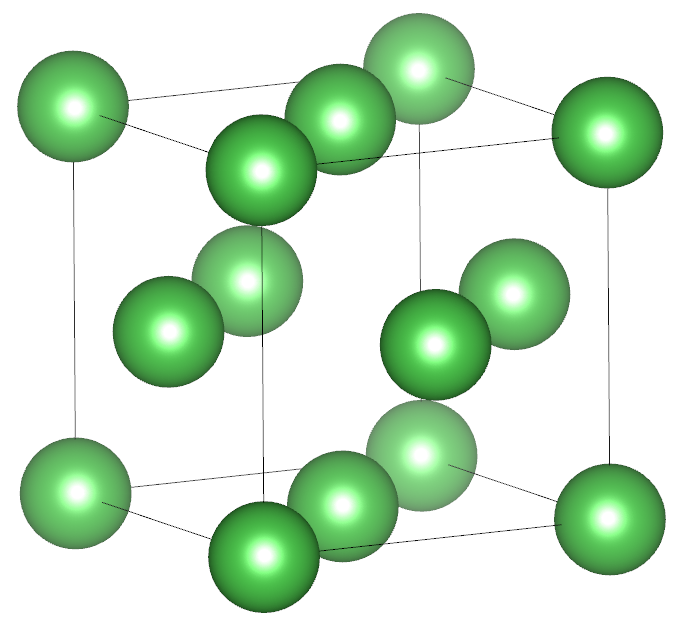}}%
  \caption{Crystal structure of the four test systems: (a). Cu-108, FCC; (b). Si-64, diamond; (c). Al-108, FCC; (d). Ni-108, FCC. }
  \label{figD2}
\end{figure}

The crystal structure of the four systems used in the tests are shown in Figure \ref{figD2}. The Cu Si, Al, and Ni system includes 108, 64, 108 and 108 atoms, respectively, and the initial training datasets are generated by running 500 AIMD steps/2fs at 300K, 500K, 300K and 500K, respectively.

\begin{table}[htbp]
  \caption{Comparison between ALKPU and DP-GEN for selecting configurations of Cu-108 system}\label{sampling}
  \centering
  \resizebox{\linewidth}{!}{
  \begin{tabular}{cccccccccc}
    \toprule
    \multicolumn{4}{c}{Exploration by MD} &  \multicolumn{3}{c}{ALKPU} & \multicolumn{3}{c}{DP-GEN}    \\
    \cmidrule(lr){1-4}  \cmidrule(lr){5-7}  \cmidrule(lr){8-10}
    Round &  Temperature  & Step  & Maximum   & Reliable  & Selected  & Discarded & Reliable  & Selected  & Discarded  \\
    \midrule
    1  & 500K & 30000  & 100 &  1710  & 1289  & 1  & 2606 & 393  & 1  \\
    2  & 500K & 30000  & 100 & 2422  & 577  & 1  & 2976 & 25  & 0  \\
    3  & 500K & 30000  & 100 & 2997  & 2   & 1  & 2989 & 11  & 0  \\
    4  & 800K & 30000  & 100 & 2547  & 452  & 1  & 1223 & 1777  & 0  \\
    5  & 800K & 30000  & 100 & 2826  & 171   & 3  & 2779 & 220  & 1  \\
    6  & 800K & 30000  & 100 & 2995  & 4  & 1  & 2885 & 115  & 0  \\
    7  & 1200K & 30000  & 100 & 2686  & 313   & 1  & 1229 & 1760  & 11  \\
    8  & 1200K & 30000 & 100 & 2892  & 105  & 3  & 2386 & 612  & 2  \\
    9  & 1200K & 30000  & 100 & 2995 & 4  & 1  & 2826 & 173  & 1  \\
    10  & 1400K & 30000  & 100 & 2998  & 2  & 0  & 2292 & 708  & 0  \\
    11  & 1400K & 30000  & 100 & 2956  & 42  & 2  & 2770 & 228  & 2  \\
    12  & 1400K & 30000  & 100 & 2990  & 9  & 1  & 2905 & 95  & 0  \\
    \bottomrule
  \end{tabular}}
\end{table}

Now we show the test results for ALKPU and the comparison with DP-GEN on the Cu system. During the exploration step, MD is run with a 2fs time interval using energies and forces predicted by the temporary trained model, and $10\%$ configurations are uniformly chosen from this MD trajectory. Note that the melting point of Cu is 1357.77K, thus a large amount of extra configurations should be select as the temporary increases from 300K to 1400K. Table \ref{sampling} compares the performance of ALKPU and DP-GEN for selecting representative configurations. For the selection step, the parameter \texttt{Maximum} is used to randomly select \texttt{Maximum} points from the selected ones if the number of selected points exceeds \texttt{Maximum}. DP-GEN also has two thresholds $(\sigma_0, \sigma_1)$ similar to $(c_0, c_1)$ to judge excessively low or high empirical variances, and here we set them to $(0.02, 0.2)$ as the default setting ($0.05, 0.15$) can not work well. For both ALKPU and DP-GEN, the number of selected configurations is consistent with the increase of temperature. For example, at the 1-st, 4-th and 7-th round, a large amount of new configuration are selected and labelled, this is because the temporary trained model's prediction accuracy is poor since it does not cover the configuration space at higher temperatures. After completing the 11-th round of active learning, the model gradually covers the configuration space of interest, leading to a decrease in the number of selected configurations. 

\begin{figure}[htbp]
  \centering
  \subcaptionbox{ \label{2a}}{\includegraphics[width=0.35\textwidth]{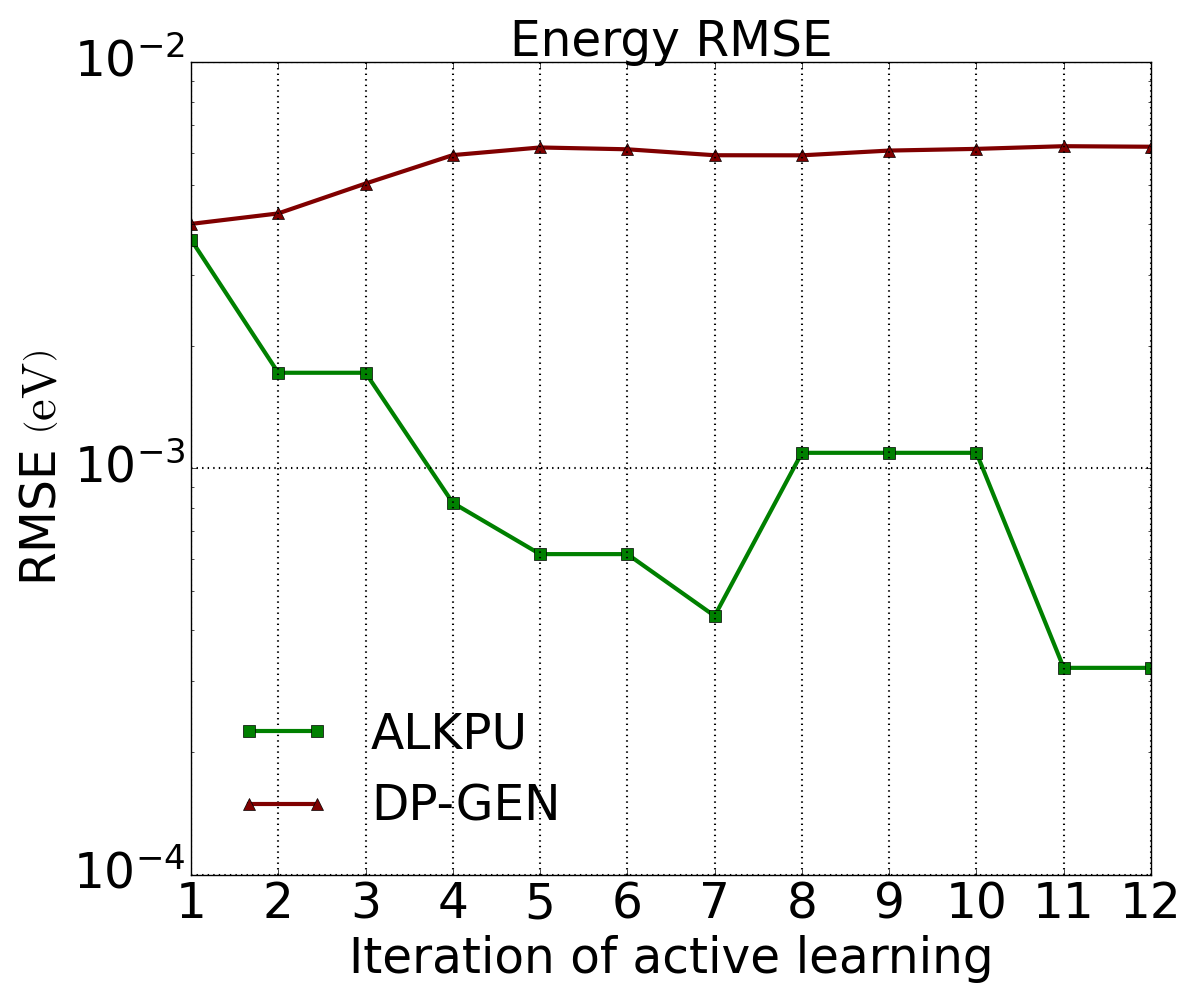}} \hspace{5mm}
  \subcaptionbox{ \label{2b}}{\includegraphics[width=0.35\textwidth]{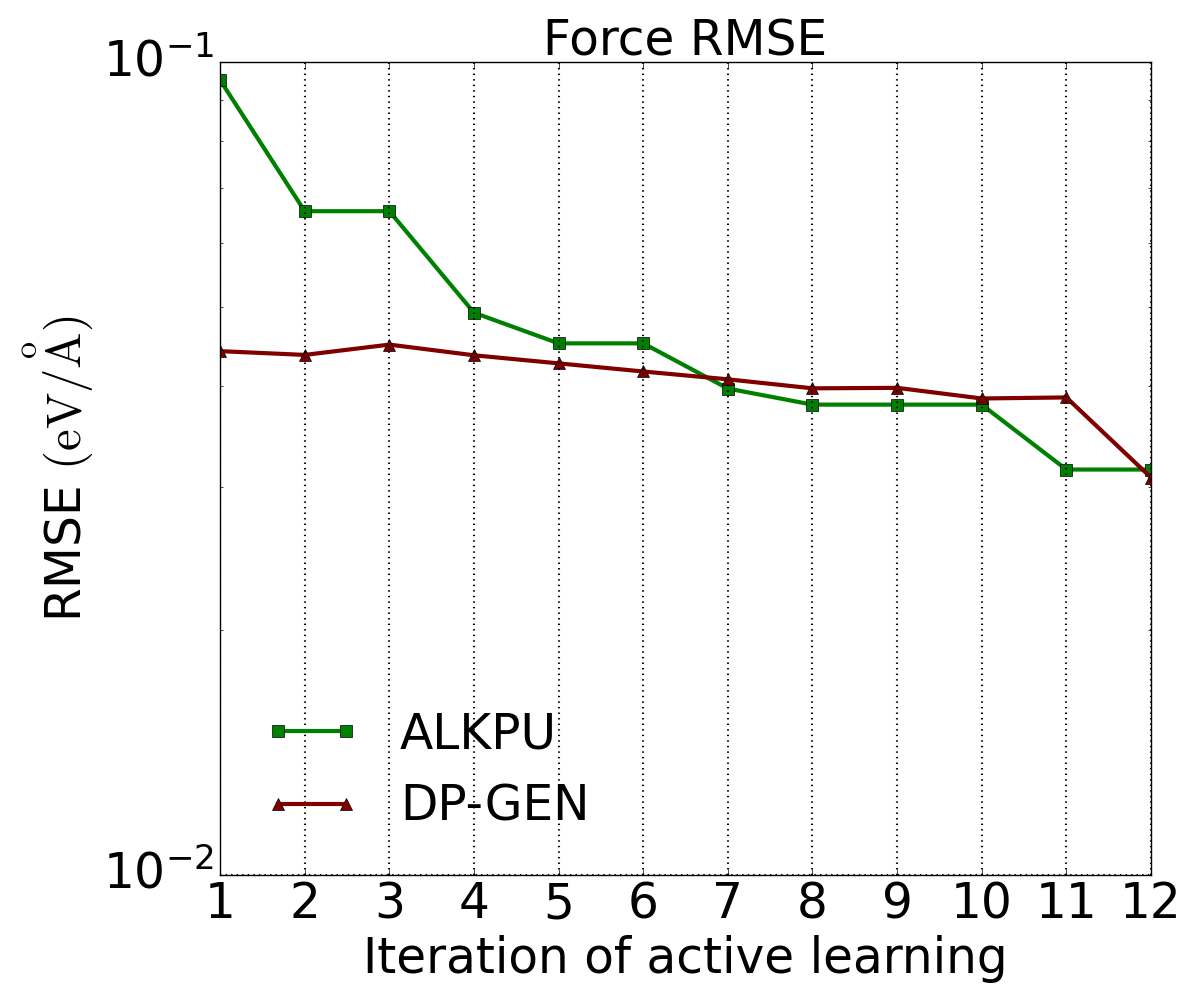}}%
  \vfill
  \subcaptionbox{ \label{2c}}{\includegraphics[width=0.35\textwidth]{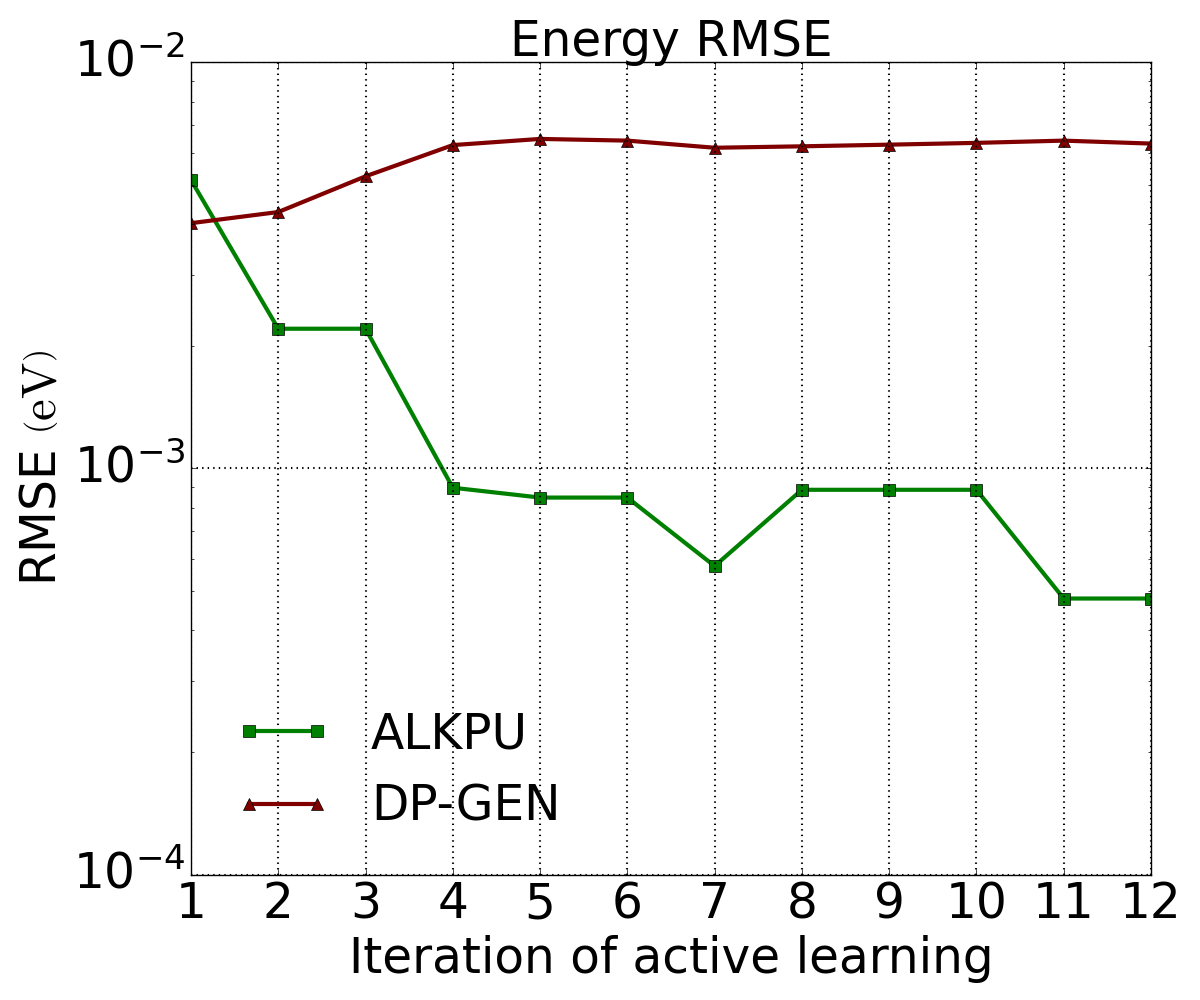}} \hspace{5mm}
  \subcaptionbox{ \label{2d}}{\includegraphics[width=0.35\textwidth]{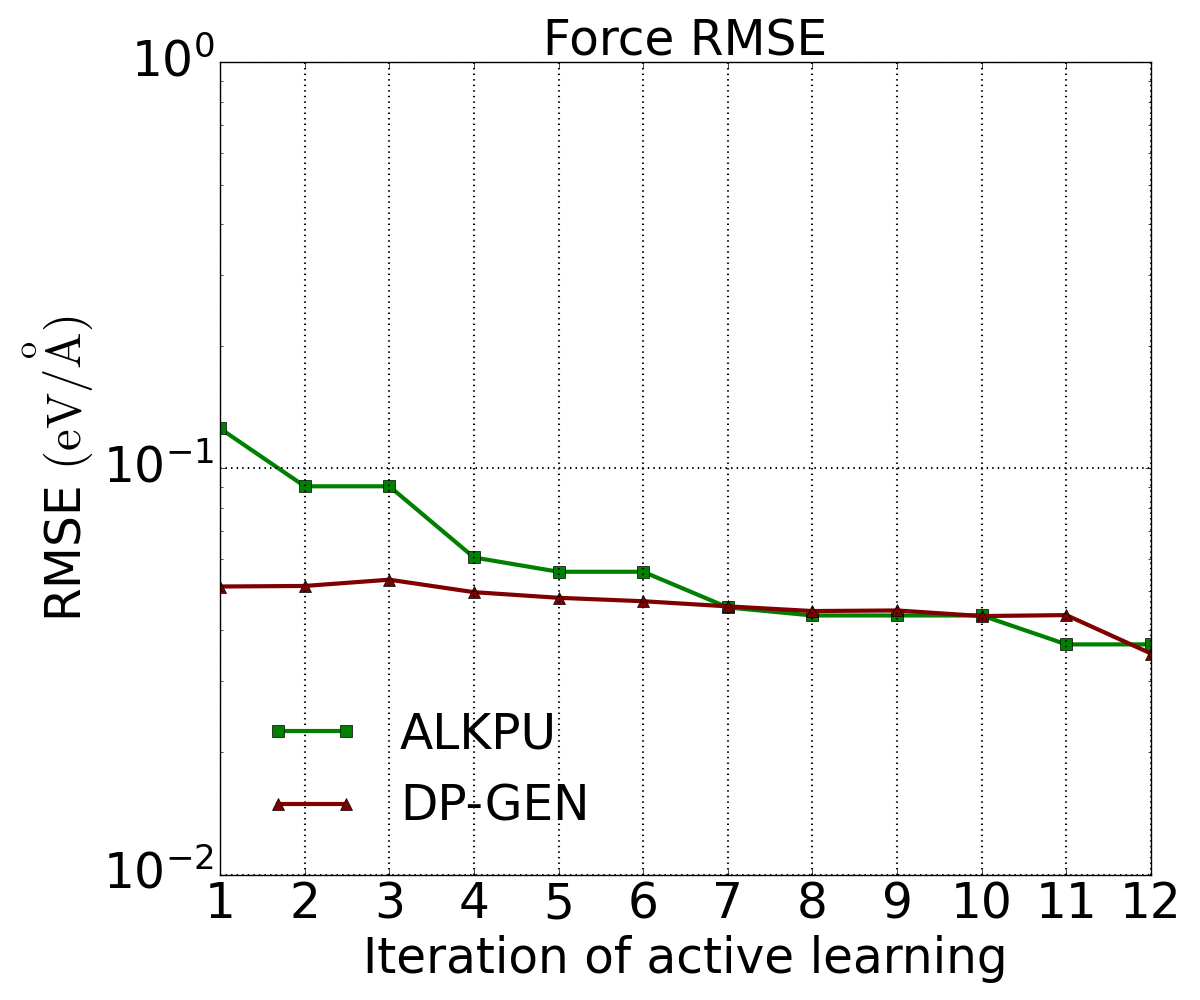}}%
  \caption{Convergence of predicted energy and force on two validation sets as the active learning proceeds, Cu-108 system: (a). Energy error, 1200K ; (b). Force error, 1200K ; (c). Energy error, 1400K; (d). Force error, 1400K. }
  \label{fig2}
\end{figure}

\begin{figure}[htbp]
  \centering
  \subcaptionbox{\label{3a}}{\includegraphics[width=0.32\textwidth]{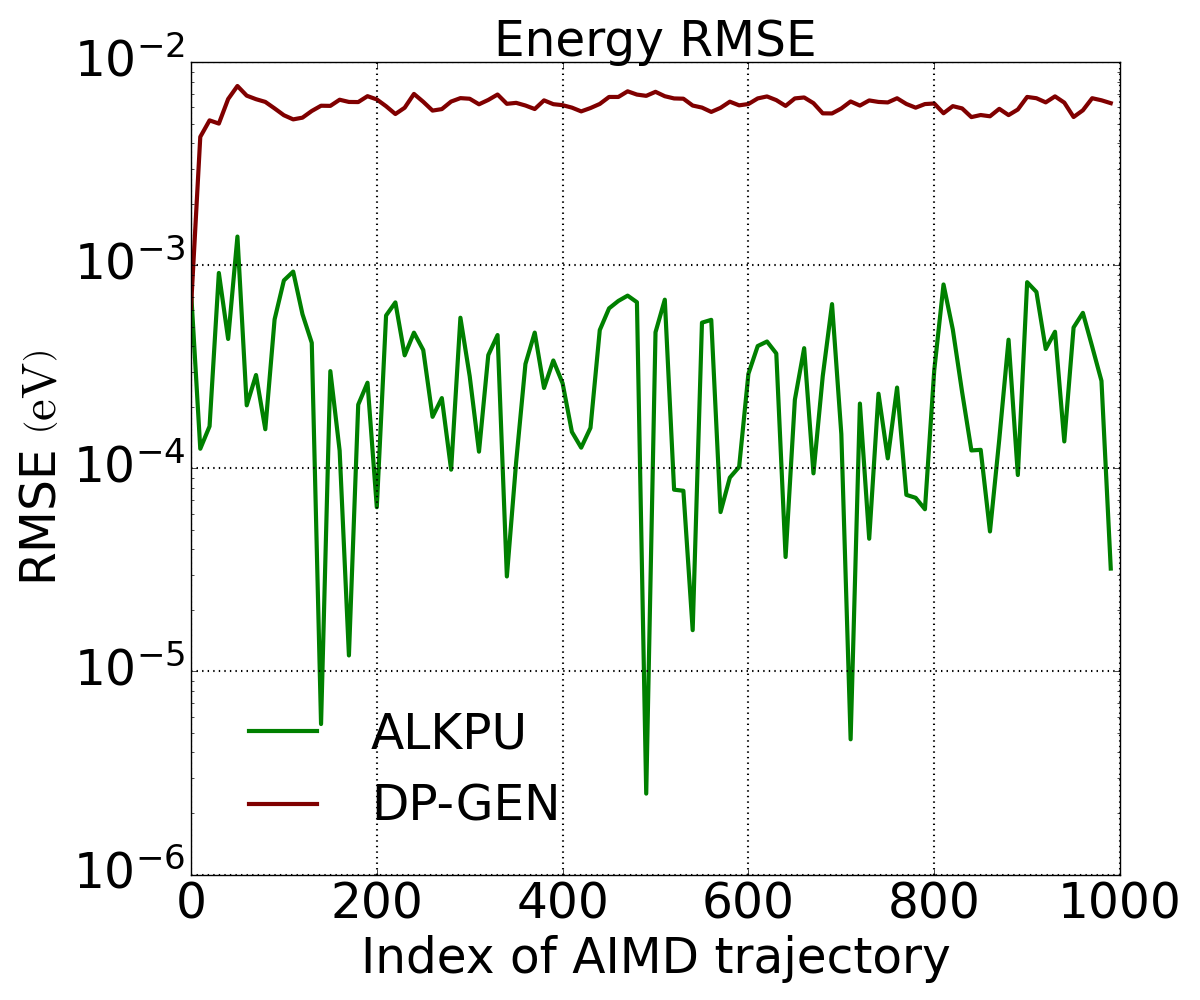}} \hspace{5mm}
  \subcaptionbox{\label{3b}}{\includegraphics[width=0.32\textwidth]{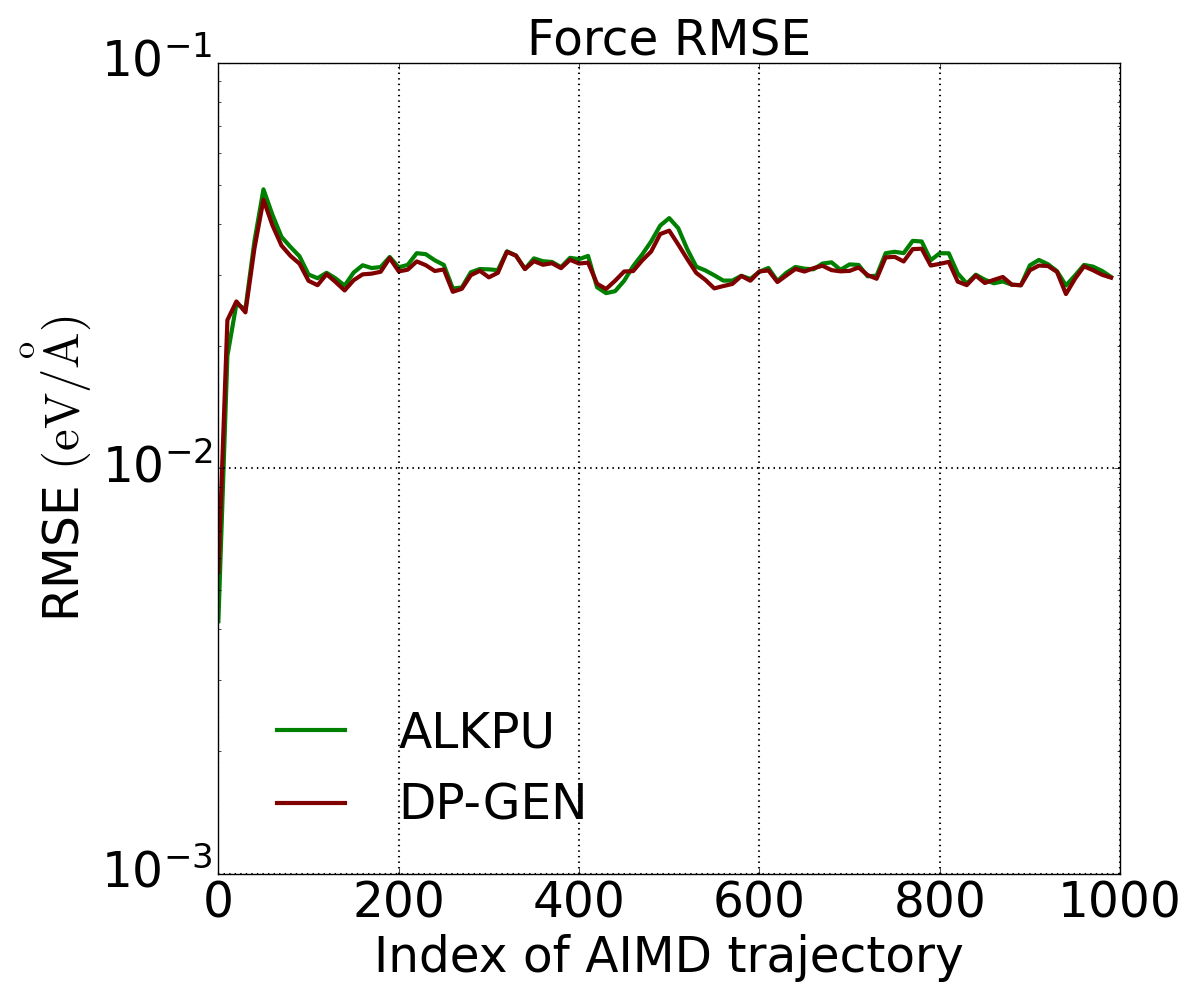}}%
  \vfill
  \subcaptionbox{\label{3c}}{\includegraphics[width=0.32\textwidth]{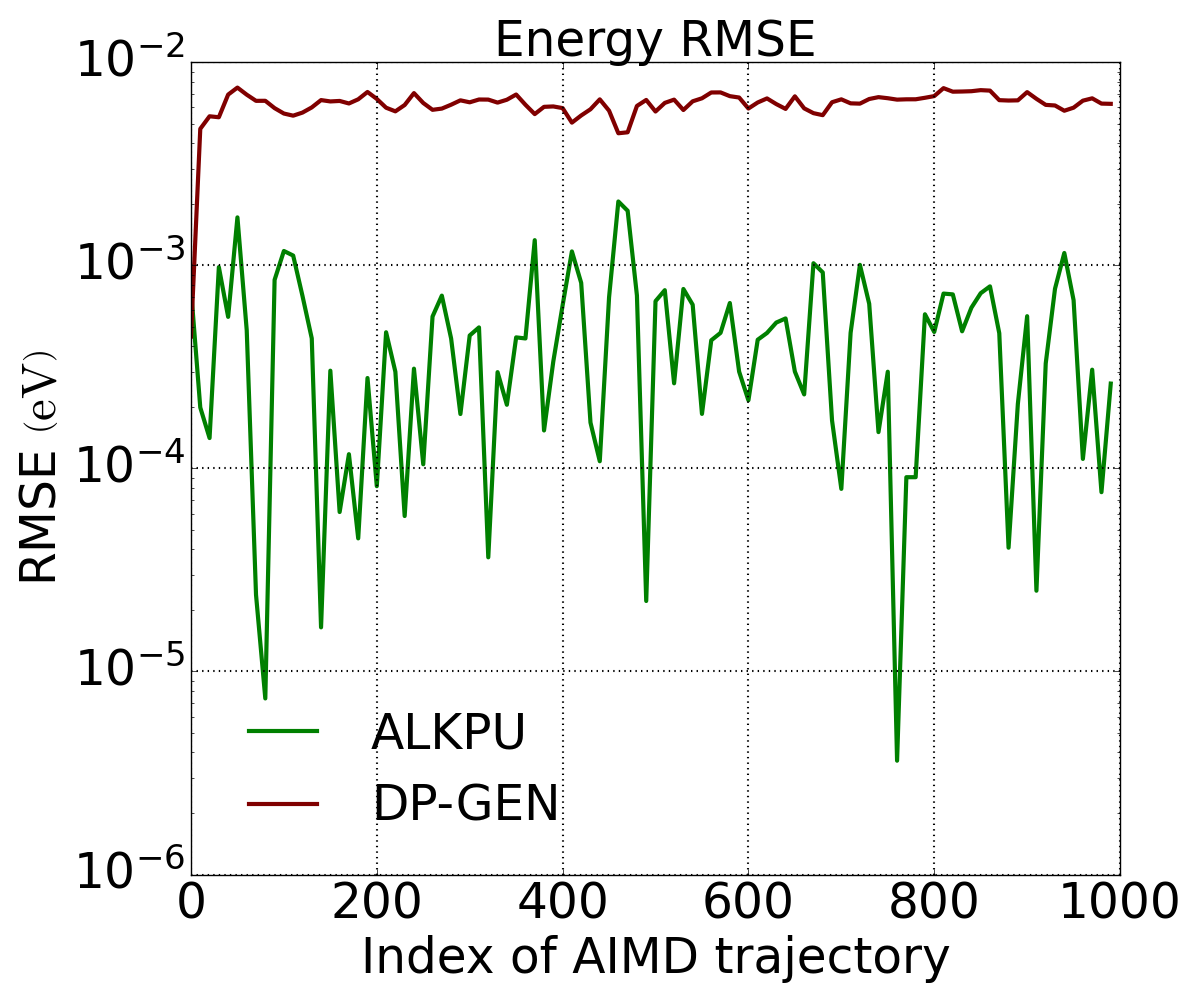}} \hspace{5mm}
  \subcaptionbox{\label{3d}}{\includegraphics[width=0.32\textwidth]{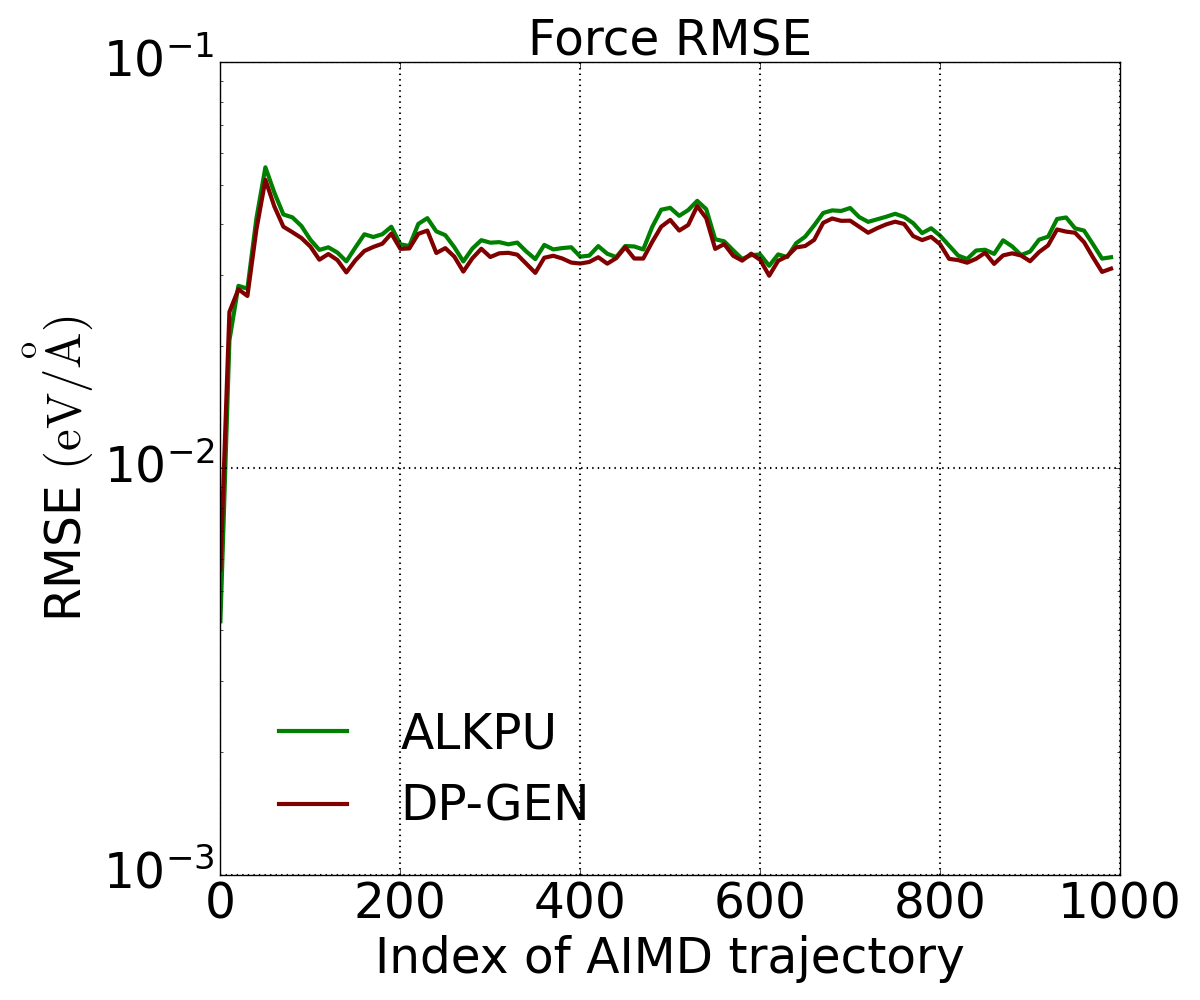}}%
  \caption{RMSE of energy and force at each configuration in the two validation sets after the active learning has finished, Cu-108 system: (a). Energy RMSE, 1200K ; (b). Force RMSE, 1200K ; (c). Energy RMSE, 1400K; (d). Force RMSE, 1400K.}
  \label{fig3}
\end{figure}

For the above two active learning procedures, Figure \ref{fig2} displays the convergence behaviors of ALKPU and DP-GEN on two validation sets as the number of active learning rounds increases. The two validation sets are generated by DFT calculations of 1000 configurations in the AIMD trajectories at 1200K and 1400K, respectively, and the convergence is measured by the average of 1000 root mean square errors (RMSE) for the 1000 configurations. As the active learning proceeds, the error gradually deceases, this is because new configurations are added to enlarge the training dataset and the configuration spaces at 1000K and 1400K are gradually covered by the trained model. For both ALKPU and DP-GEN, the errors of energy and force eventually decrease at the level of \textit{ab initio} accuracy. Note that for DP-GEN, the energy error deceases to the \textit{ab initio} accuracy at the first round, and the decreasing trend is not very obvious afterwards. For a more detailed comparison, Figure \ref{fig3} displays the RMSE of energy and force predicted by the final trained model for all the 1000 configurations of AIMD trajectories in the 1200K and 1400K validation sets. While both methods eventually achieve \textit{ab initio} accuracy, the ALKPU trained model predicts energy with slightly smaller RMSE than DP-GEN. These comparisons demonstrate that ALKPU can effectively cover the configurations space of interest during model training. Considering that ALKPU only trains one model while DP-GEN trains four models, ALKPU has advantages in terms of reducing training time and computational resource overheads.

\begin{table}[htbp]
  \caption{Time consuming for one step active learning, Cu-108 system.}\label{tab:time_count_iter}
  \centering
  \footnotesize
  \setlength{\tabcolsep}{15pt}
  \renewcommand{\arraystretch}{1.5}
  \begin{tabular}{lcccccccc}
      \toprule
        Computation & Time consuming \\ \hline
        DP-GEN training  & 0.66h (1 batchsize, 400 epochs) \\ \midrule
        ALKPU training  & 0.34h (128 batchsize, 100 epochs) \\ \midrule
        KPU computation & 0.28h (3000 configurations, 7 concurrent threads) \\ \midrule
        MD & 0.4h (30000 steps/2fs, 32 CPUs) \\ \midrule
        Labeling & 9.2h (single A100 GPU with 4 jobs, k-point (1,1,1)) \\ 
        \bottomrule
  \end{tabular}
\end{table}

Table \ref{tab:time_count_iter} summarizes the time cost statistics for one step DP-GEN and ALKPU active learning methods for the Cu-108 system. For comparison, the training was run on a single GPU. In ALKPU, the model was trained by RLEKF with a batchsize of 128, while in DP-GEN the model was trained by Adam and the batchsize is one by default of DeePMD-kit. Compared to DP-GEN, ALKPU takes more time to calculate KPU, this is due to the computation of back-propagation to get the gradient of each atomic force. One can randomly select much less atoms to calculating $\mathrm{KPU_{force}}$ to reduce the computational cost of back-propagation. Additionally, since KPU of different configurations can be computed simultaneously, this task can be divided into multiple tasks, further reducing the time required for this part. As ALKPU only needs to train one model while DP-GEN needs to train four models, the overall computational resource overhead of ALKPU is still less than that of DP-GEN.

\begin{figure}[htbp]
  \centering
  \subcaptionbox{ \label{D1a} 800K}{\includegraphics[width=0.3\textwidth]{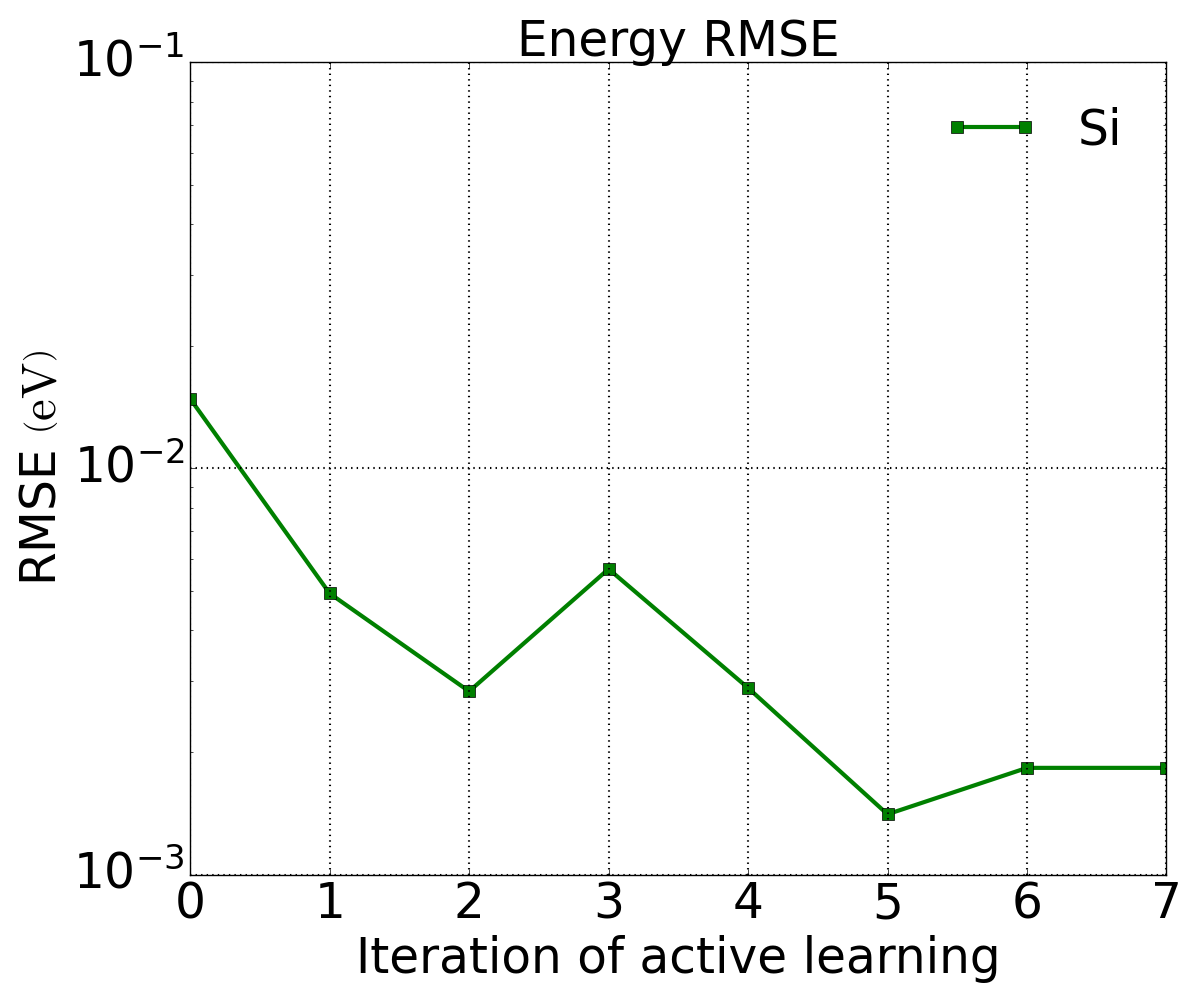}}
  \hspace{5mm}
  \subcaptionbox{ \label{D1b} 800K}{\includegraphics[width=0.3\textwidth]{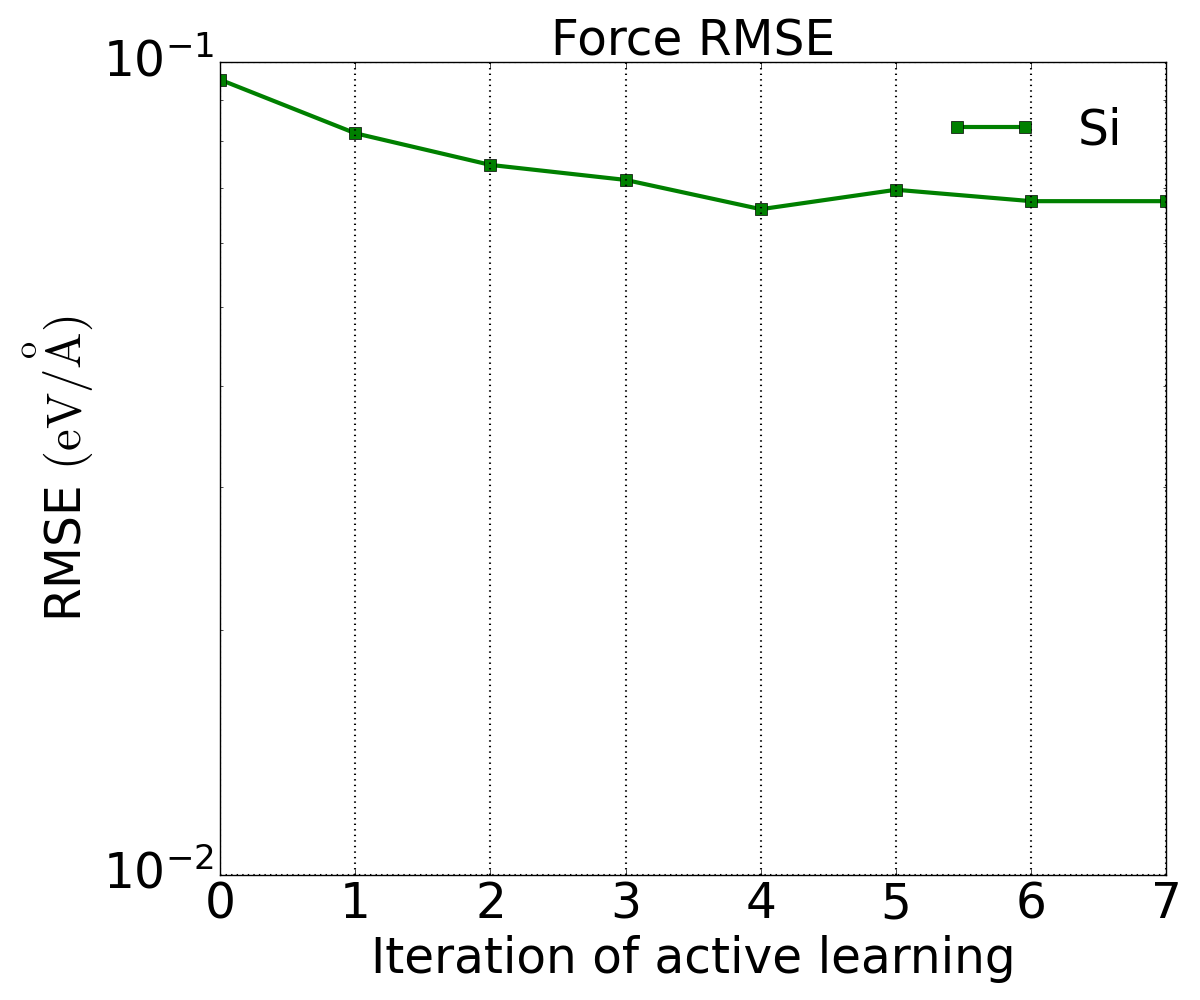}}%
  \vfill
  \subcaptionbox{ \label{D1c} 1100K}{\includegraphics[width=0.3\textwidth]{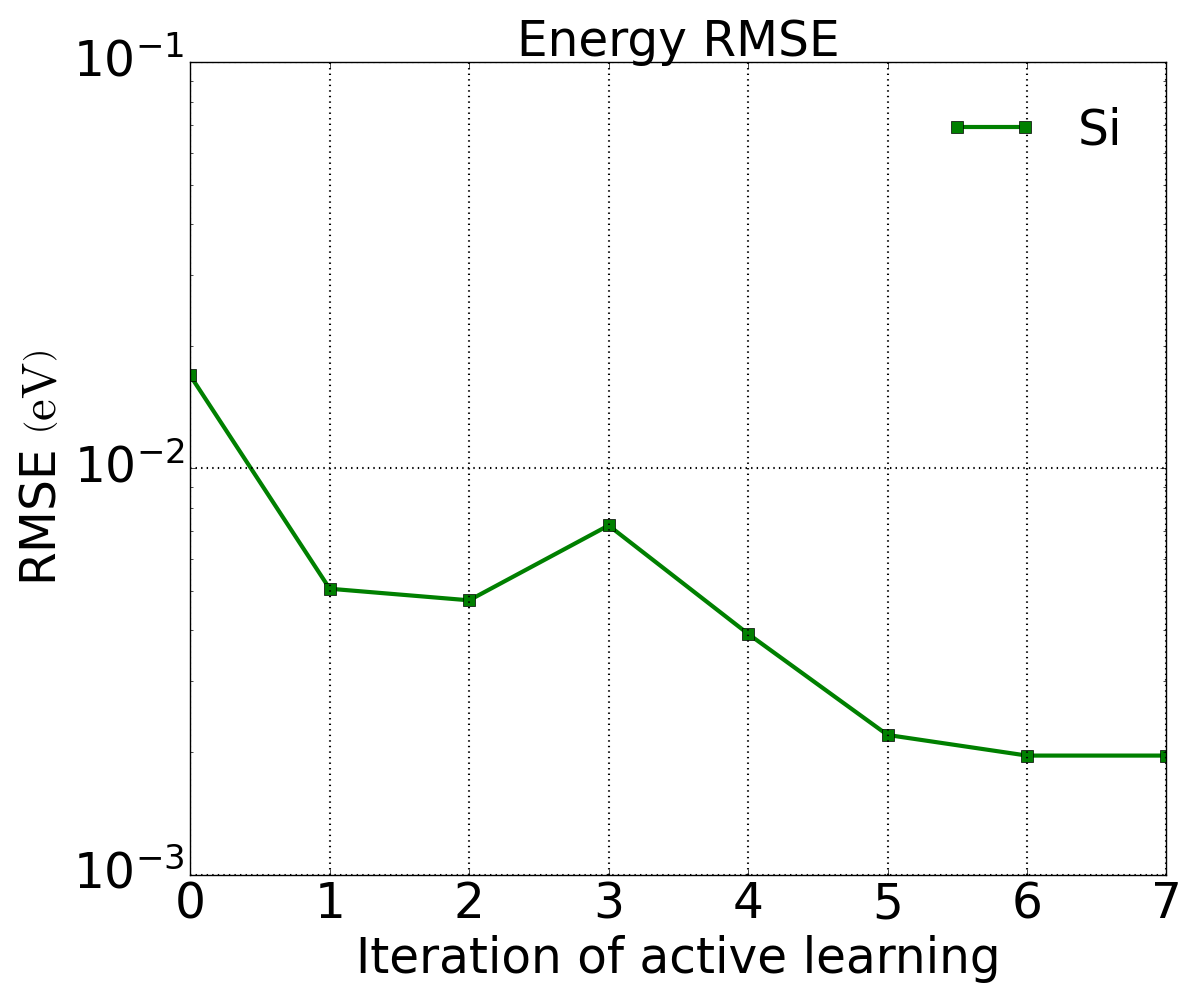}}
  \hspace{5mm}
  \subcaptionbox{ \label{D1d} 1100K}{\includegraphics[width=0.3\textwidth]{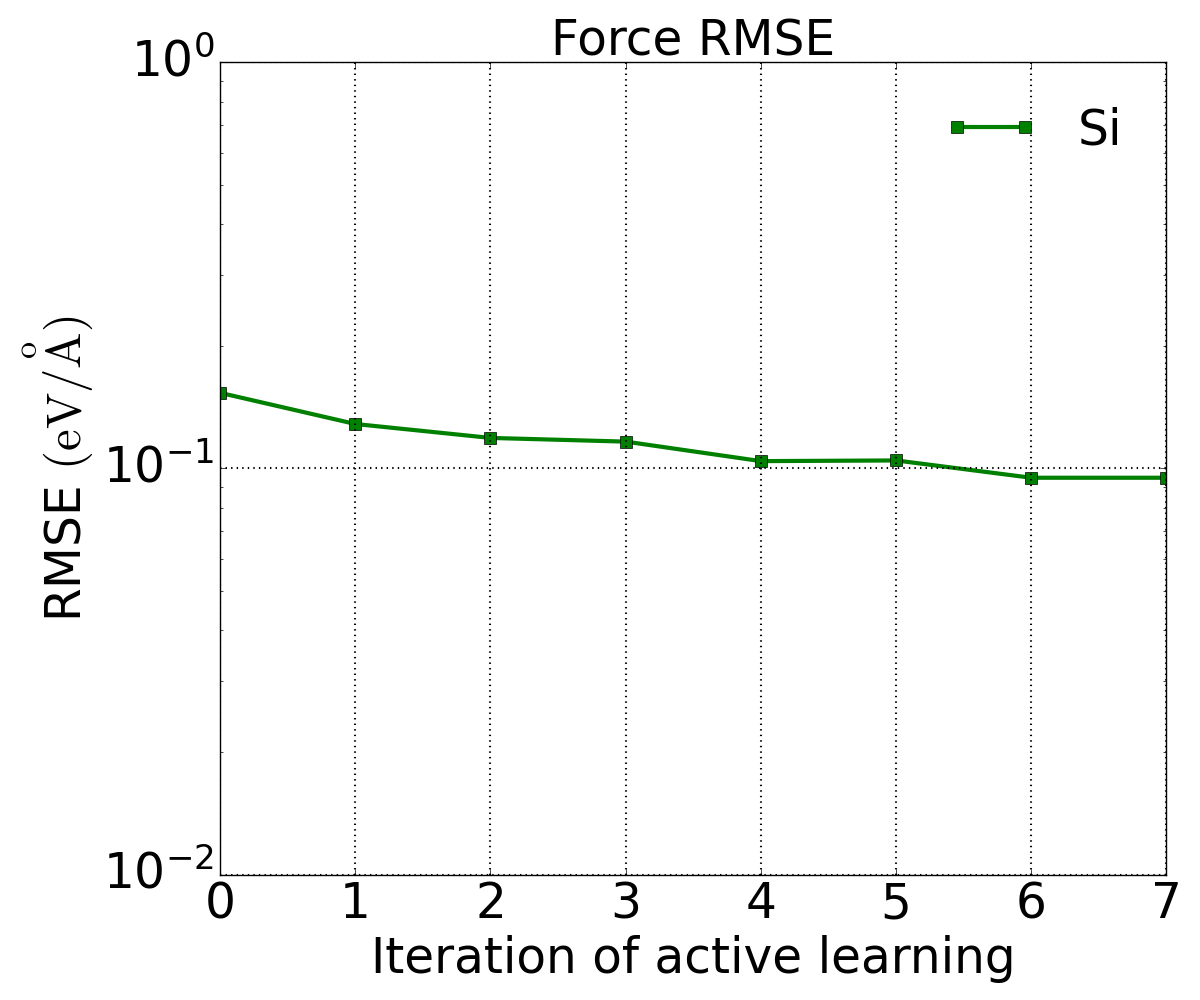}}%
  \caption{Convergence of predicted energy and force on two validation sets of as the active learning proceeds, Si-64 system: (a). Energy error, 800K; (b). Force error, 800K; (c). Energy error, 1100K; (d). Force error, 1100K. }
  \label{figD1}
\end{figure}

Figure \ref{figD1} displays the convergence behavior of ALKPU on two validation sets of the Si-64 system, where the two validation sets are generated by DFT calculations of 1000 configurations in the AIMD trajectories at 800K and 1100K, respectively, and the convergence is measured by the average of 1000 root mean square errors (RMSE) for the 1000 configurations. The convergence behavior for Al-108 and Ni-108 are similar, and thus we do not display them anymore.


\begin{figure}[htbp]
  \centering
  \subcaptionbox{ \label{4a}}{\includegraphics[width=0.4\textwidth]{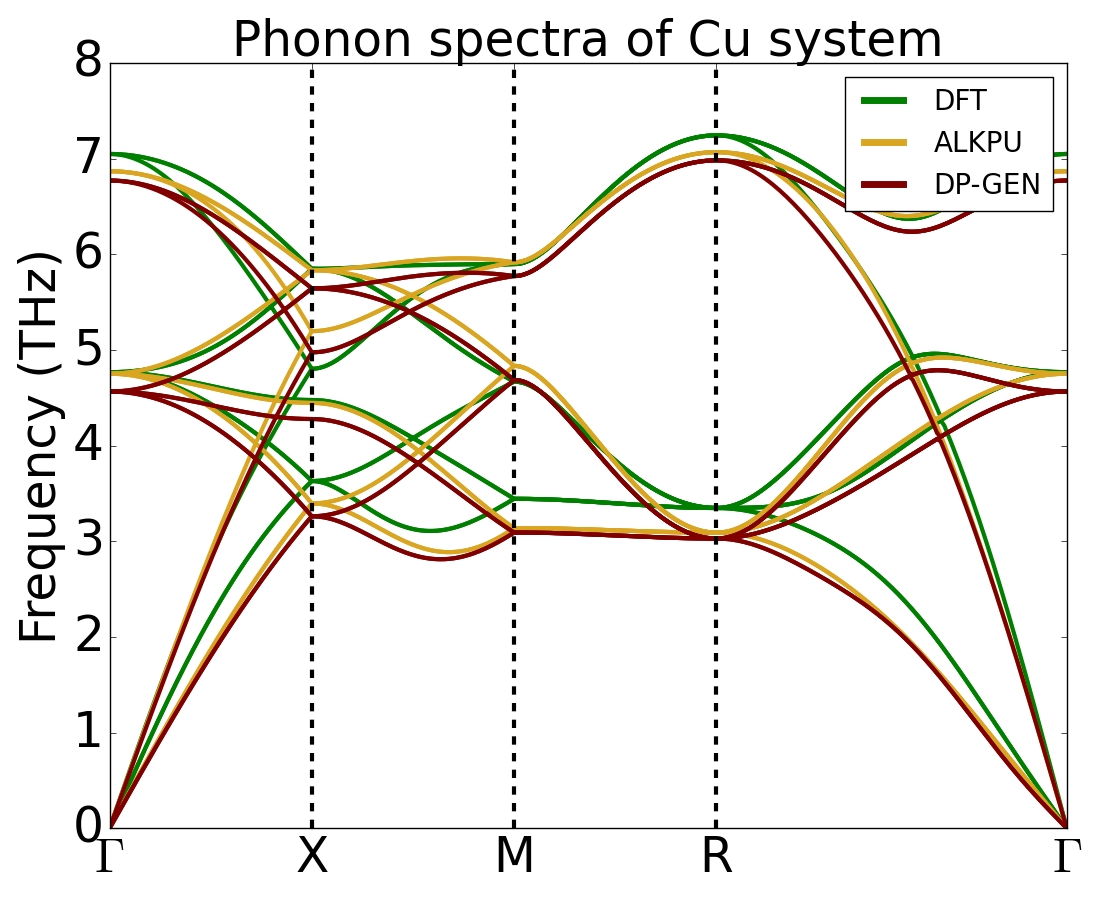}}\hspace{10mm}
  \subcaptionbox{ \label{4b}}{\includegraphics[width=0.4\textwidth]{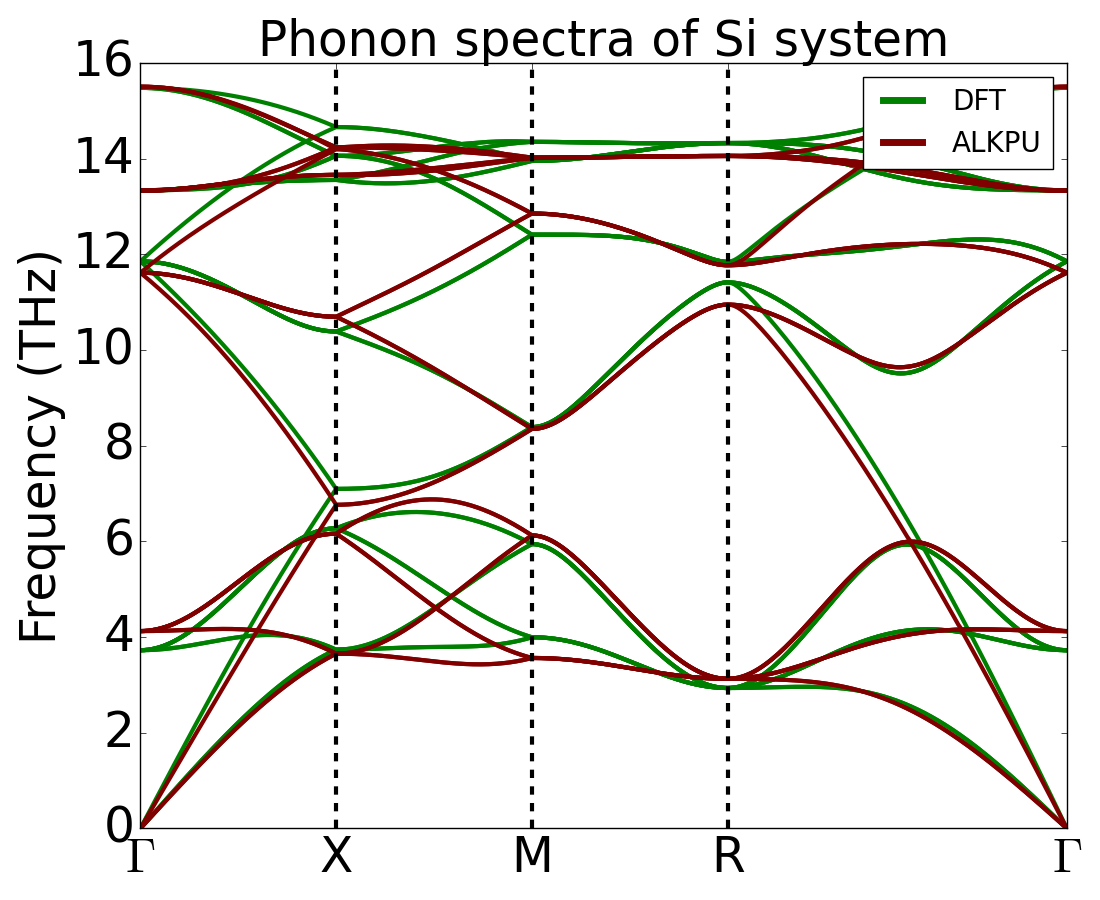}}%
  \vfill
  \subcaptionbox{ \label{4c}}{\includegraphics[width=0.4\textwidth]{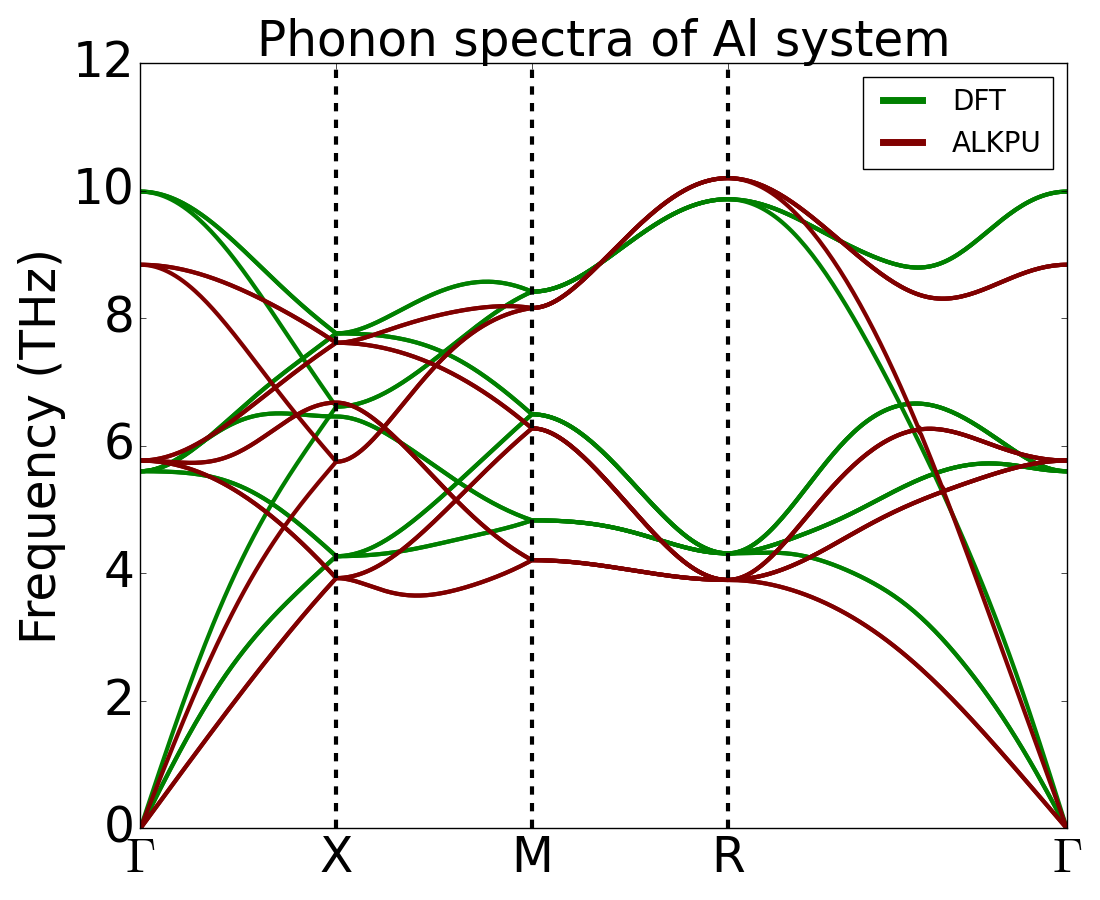}}\hspace{10mm}
  \subcaptionbox{ \label{4d} }{\includegraphics[width=0.4\textwidth]{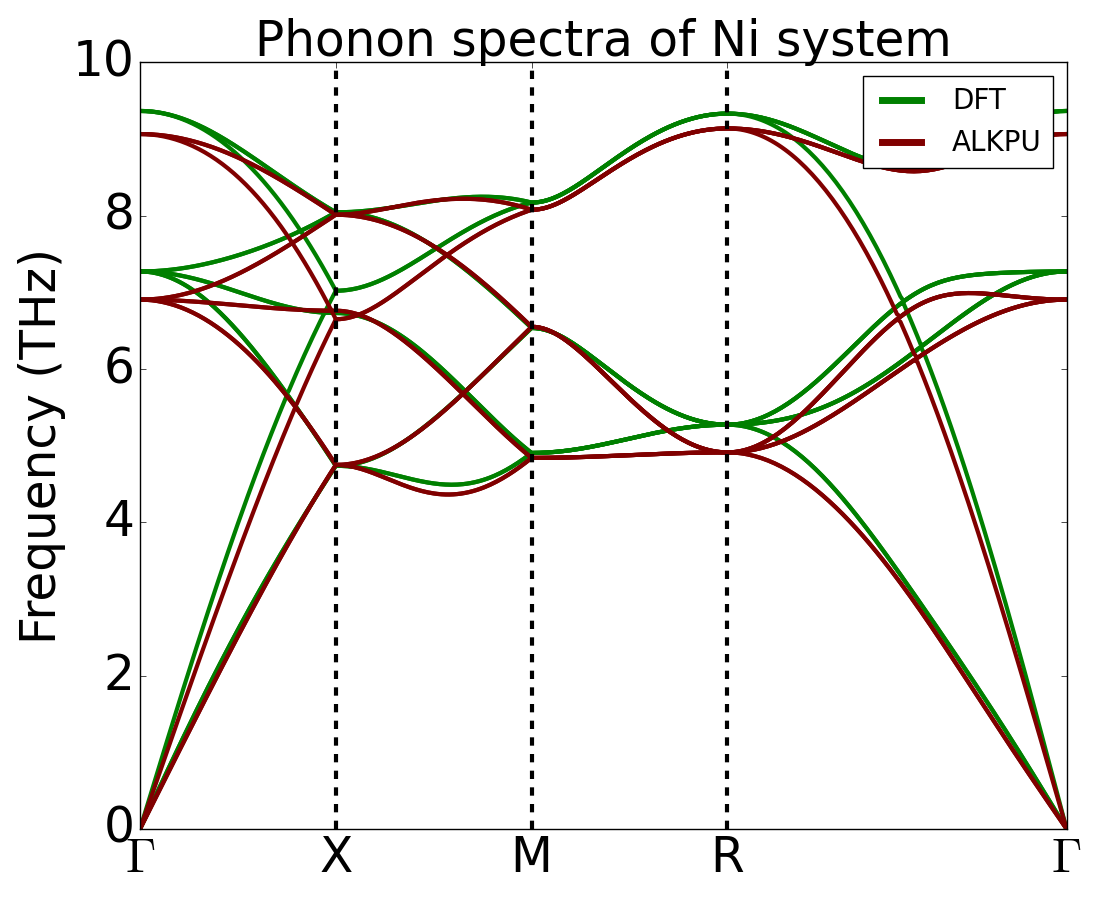}}%
  \caption{Comparison of phonon spectra calculated by DFT and DeePMD with ALKPU: (a). Cu-108, FCC; (b). Si-64, diamond; (c). Al-108, FCC; (d). Ni-108, FCC.}
  \label{fig4}
\end{figure}

In order to test the final trained model, we compare the photon spectra calculated by DFT and DeePMD with ALKPU for Cu, Si, Al and Ni systems. These test systems are chosen from the material project dataset \cite{jain_commentary_2013}. Figure \ref{fig4} shows that the trained model reproduces well the DFT results for photon spectra of the four systems (the results of Al are slightly different). We also compare the DP-GEN and ALKPU trained models for calculating Cu's photon spectra in Figure \ref{fig4} (a), and the two results are very similar. 

Overall, the test results indicate that ALKPU can effectively select representative data points to cover the configuration space during training. It has great potential for improving training efficiency to build a final usable model and can significantly accelerate MD simulations while maintaining \textit{ab initio} accuracy. In our future work, the ALKPU method will be tested on more complex systems with much larger scales to further explore its capabilities.

\section{Conclusion and outlook}\label{sec6}
To enhance the extrapolation capability and improve training efficiency for DeePMD, we have proposed ALKPU, an active learning method that can concurrently select representative configurations for labelling during model training. his approach uses the efficient RLEKF optimizer based on a Kalman filter and introduces Kalman Prediction Uncertainty (KPU) to quantify model prediction uncertainty. We developed an efficient algorithm to calculate KPU for atomic forces and an ALKPU sampling platform that concurrently explores, selects, and labels representative configurations during training. We have shown that ALKPU is essentially a locally fastest reduction procedure of model's uncertainty. We test the ALKPU method with various physical system simulations and compare it with DeePMD's original DP-GEN active learning method. Our results show that ALKPU improves training efficiency and reduces computational overhead, offering significant advantages as a new active learning approach.

Although ALKPU is designed for DeePMD, KPU is a general metric for quantifying uncertainty in NN predictions. When combined with an efficient Kalman filter-based optimizer, ALKPU can be applied to other "AI for Science" problems where model training is time-consuming and acquiring training data is computationally expensive.

\appendix
\section{The extended Kalman filter and RLEKF optimizer}\label{sec:A}
We first review the basic algorithm of the extended Kalman filter (EKF) and then derive the EKFfm algorithm. For notational simplicity, some notations in \eqref{DSD}--\eqref{eq:KFsupdate} may conflict with the notations in the main text, but the readers can easily find the differences between them.

\paragraph*{EKF}
Write a discrete stochastic dynamical system with observations as
\begin{align}\label{DSD}
  \begin{cases}
    s_{t}=f(s_{t-1}, u_t)+\xi_t, \ \ s_0 \sim \mathcal{N}(s_0, P_0),  \\
    y_{t}=h(s_{t}, u_t)+\eta_{t}, \ \ j=1,2\dots
  \end{cases}
\end{align}
where $f\in C^1(\mathbb{R}^n, \mathbb{R}^n)$, $h\in C^1(\mathbb{R}^n, \mathbb{R}^m)$ and $\xi_t\sim\mathcal{N}(0,Q_t)$, $\eta_t\sim\mathcal{N}(0,R_t)$ are noises. The first equation describes the evolution of the system with initial state $s_0$, while the second equation describes the observations, and it is usually assumed that $s_0$, $\{\xi_t\}_{t\in\mathbb{N}}$ and $\{\eta_t\}_{t\in\mathbb{N}}$ are independent. The Kalman filter is an iterative algorithm that estimates the distribution $\mathbb{P}\left(s_t|\{y_i\}_{i=1}^{i=t}\right)$ by computing its expectation and covariance \cite{kalman1960contributions,chui2017kalman}. For nonlinear mappings $f$ and $g$, the EKF applies the linear Kalman filter to a linear Gaussian system approximated by using one-order Taylor expansion, and the output is an approximation to the real distribution of $\mathbb{P}\left(s_t|\{y_i\}_{i=1}^{t}\right)$.

Given the initial state $s_0 \sim \mathcal{N}(s_0, P_0)$, for $t=1,2,\dots$, the prediction step is 
\begin{align}
F_{t-1} &= \mathrm{D}_{s} f(\hat{s}_{t-1},u_{t}),  \label{eq:transP}  \\
\bar{s}_{t} &= f(\hat{s}_{t-1},u_{t}), \label{eq:KFdefF} \\
\bar{P}_{t} &=  F_{t-1}P_{t-1}F_{t-1}^{T} + Q_{t}.
\end{align}
The correction step after observing $y_{t}$ is
\begin{align}
H_t &= \mathrm{D}_{s}h(\bar{s}_{t},u_{t}), \label{eq:defH} \\
A_t &= H_{t} \bar{P}_{t}H_{t}^{T}+R_{t}, \label{eq:defG} \\ 
K_{t} &= \bar{P}_{t}H_{t}^{T}A_{t}^{-1}, \label{eq:KFdefK} \\
P_{t} &= (I-K_{t}H_{t})\bar{P}_{t}, \label{eq:KFPupdate}  \\
d_{t} &=  y_t - h(\bar{s}_{t},u_{t}), \label{eq:defD} \\
\hat{s}_{t} &= \bar{s}_{t} + K_{t}d_{t}. \label{eq:KFsupdate}
\end{align}
The desired outputs are the expectation $\hat{s}_t$ and covariance $P_t$.

\paragraph*{Derivation of EKFfm}
Notations in this derivation follow those in the main text. By the EKF algorithm, we can estimate the state $\boldsymbol{v}_t$ and then let $\hat{\boldsymbol{w}}_t=\alpha_{t}^{-1}\hat{\boldsymbol{v}}_t$. First, we have
$\bar{\boldsymbol{v}}_{t}=\lambda_{t}^{-1/2}\hat{\boldsymbol{v}}_{t-1}$ and $F_{t-1}=\lambda_{t}^{-1/2}I$, and thus $\bar{P}_{t}=\lambda_{t}^{-1}P_{t-1}$. Then by \eqref{eq:defH} we have
\[\tilde{H}_t := \mathrm{D}_{\boldsymbol{v}}h(\alpha_{t}^{-1}\bar{\boldsymbol{v}}_{t},x_{t})=
\alpha_{t}^{-1}\mathrm{D}_{\boldsymbol{w}}h(\hat{\boldsymbol{w}}_{t-1},x_{t})=\alpha_{t}^{-1}H_t\]
since $\alpha_{t}^{-1}\bar{\boldsymbol{v}}_{t}=\alpha_{t}^{-1}\lambda_{t}^{-1/2}\alpha_{t-1}\hat{\boldsymbol{w}}_{t-1}=\hat{\boldsymbol{w}}_{t-1}$. By \eqref{eq:defG} we have
\[\tilde{A}_t := \tilde{H}_{t}\bar{P}_{t}\tilde{H}_{t}^{T}+R_t = 
\alpha_{t}^{-2}(H_{t}\bar{P}_{t}H_{t}^{T}+\alpha_{t}^{2}R_t) = \alpha_{t}^{-2}A_t .\]
By \eqref{eq:KFdefK} we have
\[\tilde{K}_t = \bar{P}_{t}\tilde{H}_{t}^{T}\tilde{A}_{t}^{-1}
= \alpha_{t}(\bar{P}_{t}H_{t}^{-1}A_{t}^{-1}) = \alpha_{t}K_t .\]
Therefore, the estimated covariance of $\boldsymbol{v}_t|\{y_i\}_{i=1}^{t}$ is
\[P_t = (I-\tilde{K}_t\tilde{H}_t)\bar{P}_{t} = (I-K_t H_t)\bar{P}_t .\]
By \eqref{eq:defD} we have
\[d_t = y_t-h(\alpha_{t}^{-1}\bar{\boldsymbol{v}}_{t}, x_t)=y_t-h(\hat{\boldsymbol{w}}_{t-1}, x_t) .\]
Therefore, the estimated expectation of $\boldsymbol{v}_t|\{y_i\}_{i=1}^{t}$ is
\[\hat{\boldsymbol{v}}_t =\bar{\boldsymbol{v}}_{t}+\tilde{K}_t d_t = \lambda_{t}^{-1/2}\hat{\boldsymbol{v}}_{t-1}+\alpha_{t}K_{t}d_t,\]
which finally leads to
\[\hat{\boldsymbol{w}}_t = \alpha_{t}^{-1}\hat{\boldsymbol{v}}_t = \alpha_{t}^{-1}\lambda_{t}^{-1/2}\alpha_{t-1}\hat{\boldsymbol{w}}_{t-1}+K_{t}d_t 
= \hat{\boldsymbol{w}}_{t-1}+K_{t}d_t .\]
One can check the correspondance of computations in the above and in Algorithm \ref{EKFfm}. From the derivation we known that the estimated covariance of $\boldsymbol{w}$ at the $t$-th step is 
\[\hat{P}_t = \alpha_{t}^{-2}P_t .\]

\paragraph*{RLEKF optimizer}
The RLEKF approximates dense matrices in EKFfm with sparse diagonal block ones to reduce the computational cost of matrix-matrix/vector multiplications. Technically, all the layers of the NN are reorganized by splitting big and gathering adjacent small layers. 

The layerwise weight-updating procedure of RLEKF for training DeePMD in shown in Algorithm \ref{RLEKF2}. In lines 1--6, we first transform a vector $Y$ into a scalar such that the prediction of the NN model has dimension 1. \texttt{split} is a function for splitting and reorganizing layers, while \texttt{gather} is a function for gathering all reorganized layers together. Some computations have been modified to make the practical implementation more robust; for details see \cite{hu2022rlekf}.

\begin{algorithm}
  \caption{\texttt{RLEKF}$(\boldsymbol{Y}, \boldsymbol{Y}_{\text{DFT}}, \boldsymbol{w}^{in}, P^{in}, \lambda^{in})$}\label{RLEKF2}
  \begin{algorithmic}[1]
    \Require $\boldsymbol{w}^{in}, P^{in}, \lambda^{in}$
    \For{$j =1,2,\dots, \mathtt{length}(\boldsymbol{Y})$} 
      \If{$Y_{j} \geq Y_{j,\text{DFT}}$}
        \State $Y_{j}=-Y_{j}$; \ \ $Y_{j,\text{DFT}}=-Y_{j,\text{DFT}}$
      \EndIf
    \EndFor
    \State $Y =\sum_{j}Y_{j}/\mathtt{length}(\boldsymbol{Y})$; \ \
    $\mathrm{er} =\sum_{j}(Y_{j,\text{DFT}}-Y_{j})/\mathtt{length}(\boldsymbol{Y})$
    \State $H =\nabla_{\boldsymbol{w}}Y\mid_{w=w^{in}}$  \Comment{Backward propagation}
    \State $\alpha = 1/(\lambda^{in}L + H^{T}P^{in}H)$  \Comment{$L$ is the number of reorganized layers}
    \State $\{P_{1}^{in},\dots,P_{L}^{in}\}=\mathtt{split}(P^{in}); \ \ \{\boldsymbol{w}_{1}^{in},\dots,\boldsymbol{w}_{L}^{in}\}=\mathtt{split}(P^{in}); \ \ 
    \{H_{1},\dots,H_{L}\}=\mathtt{split}(H)$ \Comment{Reorganizing by layers}
    \For{$l =1,2,\dots, L$}
    \State $K = P^{in}_{l}H_{l}$  
    \State $P^{out}_{l}= (P^{in}_{l}-\alpha KK^{T})/\lambda^{in} $  
    \State $\boldsymbol{w}^{out}_{l} = \boldsymbol{w}^{in}_{l}+ \alpha\mathrm{er}KY$
    \EndFor 
    \State $\boldsymbol{w}^{out} = \mathtt{gather}(\boldsymbol{w}^{out}_{1},\dots, \boldsymbol{w}^{out}_{L}) $
    \State $P^{out} =\mathtt{gather}(P^{out}_{1},\dots, P^{out}_{L}) $  \Comment{Gathering all layers}
    \State $\lambda^{out} = \lambda^{in} \nu +1- \nu $ 
    \Ensure $\boldsymbol{w}^{out}, P^{out}, \lambda^{out} $ 
  \end{algorithmic}
\end{algorithm}

The overall structure of the training procedure is outlined in Algorithm \ref{RLEKF}. The algorithm begins by initializing $\hat{\boldsymbol{w}}_0$, $P_0$ and $\lambda_1$. At the $t$-th step, the energy process fits the predicted total energy to the labeled energy of configuration $\mathcal{R}^{(t)}$ by nonlinear NN mapping $h_{E}$, while the force process fits the predicted atom forces to the labeled forces by nonlinear mapping $h_{F}$ that is the gradient of $h_{E}$ with respect to coordinates. During the force process, a set of $n_a$ atoms is randomly chosen from the configuration and the corresponding force vectors are concatenated into a scalar (i.e. lines 1--6 of Algorithm \ref{RLEKF2}). This process (lines 4--5) is repeated $n_F$ times in one step. In the practical implementation we set $n_a=6$ and $n_F=4$. 

\begin{algorithm}
  \caption{High-level structure of training by RLEKF}\label{RLEKF}
  \begin{algorithmic}[1]
    \Require $\hat{\boldsymbol{w}}_0$, $P_0$, $\lambda_1$, $\nu$, training data $\{\mathcal{R}^{(t)}, E_{\text{DFT}}^{(t)}, \{\boldsymbol{F}_{i,\text{DFT}}^{(t)}\}\}_{t\in\mathbb{N}}$
    \For{$t =1, 2,\dots,T$}
    \State $E = h_E(\hat{\boldsymbol{w}}_{t-1}, \mathcal{R}^{(t)})$  \Comment{Predicted total energy}
    \State $\hat{\boldsymbol{w}}_{+}, P_{+}, \lambda_{+} =\mathtt{RLEKF}(E, E_{\text{DFT}}^{(t)}, \hat{\boldsymbol{w}}_{t-1}, \mathbf{P}_{t-1}, \lambda_{t})$ \Comment{Training by energy}
    \State $\{\boldsymbol{F}_{i}\} = h_{F}(\hat{\boldsymbol{w}}_{+}, \mathcal{R}^{(t)})$  \Comment{Predicted forces}
    \State $\hat{\boldsymbol{w}}_{t}, P_{t}, \lambda_{t+1} =\mathtt{RLEKF}(\{\boldsymbol{F}_{i}\}, \{\boldsymbol{F}_{i,\text{DFT}}^{(t)}\}, w_{+}, \mathbf{P}_{+}, \lambda_{+})$  \Comment{Training by forces}
    \EndFor
    \Ensure $\hat{\boldsymbol{w}}_T$, $P_T$, $\lambda_{T+1}$
  \end{algorithmic}
\end{algorithm}

\section{Shannon entropy of the multivariate Gaussian distribution}\label{sec:B}
A vector-valued random variable $X=[X_1, \dots, X_n]^T\in\mathbb{R}^n$ is said to have a multivariate Gaussian distribution if its probability density function is given by
\begin{equation}
  p(x) = \frac{1}{(2\pi)^{n/2}(\det\Sigma)^{1/2}}\exp{\Big(-\frac{1}{2}(x-\mu)^{T}\Sigma^{-1}(x-\mu)\Big)},
\end{equation}
with the expectation (or mean) $\mu\in\mathbb{R}^{n}$ and covariance $\Sigma\in\mathbb{R}^{n\times n}$ where $\Sigma$ is semi-definite. We write it as $X\sim\mathcal{N}(\mu,\Sigma)$.

The Shannon entropy of $X$ can be calculated as follows. First,
\begin{align*}
  H(X) &= -\int_{\mathbb{R}^n}p(x)\ln p(x)\mathrm{d}x \\
       &= -\int_{\mathbb{R}^n}[\ln((2\pi)^{n/2}(\det\Sigma)^{1/2})+
       \frac{1}{2}(x-\mu)^{T}\Sigma^{-1}(x-\mu)]p(x)\mathrm{d}x \\
       &= \frac{n}{2}\ln(2\pi)+\frac{1}{2}\ln(\det \Sigma)+\frac{1}{2}\mathbb{E}_{x\sim p(x)}[(x-\mu)^{T}\Sigma^{-1}(x-\mu)].
\end{align*}
Note that $(x-\mu)^{T}\Sigma^{-1}(x-\mu)$ is a scalar. Thus we have
\begin{align*}
  \mathbb{E}_{x\sim p(x)}[(x-\mu)^{T}\Sigma^{-1}(x-\mu)]
  &= \mathbb{E}_{x\sim p(x)}[\mathrm{tr}((x-\mu)^{T}\Sigma^{-1}(x-\mu))] \\
  &= \mathbb{E}_{x\sim p(x)}[\mathrm{tr}(\Sigma^{-1}(x-\mu)(x-\mu)^{T})] \\
  &= \mathrm{tr}(\Sigma^{-1}\mathbb{E}_{x\sim p(x)}[(x-\mu)(x-\mu)^{T}]) \\
  &= \mathrm{tr}(\Sigma^{-1}\Sigma) \\
  &= n.
\end{align*}
Therefore, we finally obtain
\begin{equation}\label{entropy}
  H(X) = \frac{n}{2}(1+\ln(2\pi))+\frac{1}{2}\ln(\det \Sigma) .
\end{equation}
Note that for the $n$ dimensional Gaussian random vectors, the Shannon entropy depends only on its covariance but not its mean.

\section{Proofs}\label{sec:C}
\begin{proof}[Proof of Theorem \ref{ONG}]
  In Algorithm \ref{EKFfm}, at the $t$-th step we have
  \begin{align*}
    P_t &= \bar{P}_{t}- \bar{P}_{t}H_{t}^{T}(H_{t}\bar{P}_{t}H_{t}^{T}+\alpha_{t}^{2}R_{t})^{-1}H_{t}\bar{P}_{t} \\
        &= (\bar{P}_{t}^{-1}+H_{t}^{T}\alpha_{t}^{-2}R_{t}^{-1}H_{t})^{-1} \\
        &= (\lambda_{t}P_{t-1}^{-1}+H_{t}^{T}\alpha_{t}^{-2}R_{t}^{-1}H_{t})^{-1},
  \end{align*}
  where the second equality is obtained by using the Sherman-Morrison-Woodbury formula, which leads to
  \begin{equation}\label{inv_P}
    P_{t}^{-1} = \lambda_{t}P_{t-1}^{-1} + \alpha_{t}^{-2}H_{t}^{T}R_{t}^{-1}H_{t} .
  \end{equation}
  
  Let $\hat{y}_t=h(\hat{\boldsymbol{w}}_{t-1},x_t)$ and $\ell_{t}(\boldsymbol{w})=\ln p(y_t|\boldsymbol{w})$. Since $y_t=h(\boldsymbol{w},x_t)+\eta_t$ with $\eta_t\sim\mathcal{N}(0,R_t)$, we have
  \begin{equation*}
    p(y_t|\hat{\boldsymbol{w}}_{t-1}) \propto \exp{\Big(-\frac{1}{2}(y_t-\hat{y}_t)^{T}R_{t}^{-1}(y_t-\hat{y}_t)\Big)},
  \end{equation*}
  and thus 
  \begin{equation}\label{ellt}
    \ell_{t}(\hat{\boldsymbol{w}}_{t-1}) = -\frac{1}{2}(y_t-\hat{y}_t)^{T}R_{t}^{-1}(y_t-\hat{y}_t) + C_1
  \end{equation}
  with $C_1$ a constant number. Notice that 
  \begin{align*}
    \nabla_{\boldsymbol{w}} \ell_{t}(\hat{\boldsymbol{w}}_{t-1})
    = H_{t}^{T}\nabla_{\hat{y}}\ell_{t}\big|_{\hat{y}_t} ,
  \end{align*}
  where $\nabla_{\hat{y}}\ell_{t}$ is the gradient with respect to the middle variable $\hat{y}$. Therefore 
  \begin{equation*}
    \mathbb{E}_{y_t\sim p(y_t|\hat{\boldsymbol{w}}_{t-1})}\Big [\nabla_{\boldsymbol{w}}\ell_{t}(\hat{\boldsymbol{w}}_{t-1})^{\otimes 2} \Big ] 
    = H_{t}^T\mathbb{E}_{y_t}\Big [\Big(\nabla_{\hat{y}}\ell_{t}\big|_{\hat{y}_t}\Big)^{\otimes 2} \Big ]H_{t}.
  \end{equation*}
  By \eqref{ellt} we have
  \begin{align*}
    \mathbb{E}_{y_t}\Big[\Big(\nabla_{\hat{y}}\ell_{t}\big|_{\hat{y}_t}\Big)^{\otimes 2}\Big]
    &= \mathbb{E}_{y_t}[R_{t}^{-1}(y_t-\hat{y}_t)] [R_{t}^{-1}(y_t-\hat{y}_t)]^T  \\
    &= R_{t}^{-1}\mathbb{E}_{y_t}[(y_t-\hat{y}_t)(y_t-\hat{y}_t)^T] R_{t}^{-1} \\
    &= R_{t}^{-1} .
  \end{align*}
  Combining with \eqref{inv_P} we obtain
  \begin{equation}
    P_{t}^{-1} = \lambda_{t}P_{t-1}^{-1} + \alpha_{t}^{-2}\mathbb{E}_{y_t\sim p(y_t|\hat{\boldsymbol{w}}_{t-1})}\Big [\nabla_{\boldsymbol{w}}\ell_{t}(\hat{\boldsymbol{w}}_{t-1})^{\otimes 2}\Big ].
  \end{equation}
  By defining $J_t=\gamma_{t}P_t^{-1}$ we get
  \[J_t = \lambda_{t}\frac{\gamma_t}{\gamma_{t-1}}J_{t-1} +
  \frac{\gamma_t}{\alpha_{t}^2}\mathbb{E}_{y_t}\Big[\nabla_{\boldsymbol{w}}\ell_{t}(\hat{\boldsymbol{w}}_{t-1})^{\otimes 2}\Big] .\]
  By letting $\lambda_{t}\frac{\gamma_t}{\gamma_{t-1}}+\frac{\gamma_t}{\alpha_{t}^2}=1$ we finally obtain \eqref{Jt}--\eqref{gamma}.
\end{proof}

\begin{proof}[Proof of Proposition \ref{limit}]
  First we prove $\lim_{t\rightarrow \infty}\lambda_{t} = 1$ and 
  \begin{equation}
    \lim_{t\rightarrow \infty}\alpha_{t}^{-2} = C_2
  \end{equation}
  with $C_2$ a positive number. We have
  \begin{align*}
    \lambda_{t}
    &= 1-\nu + \nu\lambda_{t-1} = 1-\nu + \nu(1-\nu+\nu\lambda_{t-2}) \\
    &= \cdots \\
    &= (1-\nu)(1+\nu+\cdots+\nu^{t-2}) + \nu^{t-1}\lambda_1 \\
    &= 1-\nu^{t-1} + \nu^{t-1}\lambda_1 \longrightarrow 1, \ \ t\rightarrow \infty .
  \end{align*}
  Note that $\alpha_{t}^{-2} = \lambda_1\lambda_2\cdots\lambda_t$, which leads to
  \begin{align*}
    \ln \alpha_{t}^{-2} 
    &= \sum_{i=1}^{t}\ln(1-\nu^{i-1}(1-\lambda_1)) \\
    &\geq \sum_{i=1}^{t}\frac{\nu^{i-1}(1-\lambda_1)}{\nu^{i-1}(1-\lambda_1)-1} \\
    &\geq \sum_{i=1}^{t}\frac{\nu^{i-1}(1-\lambda_1)}{-\lambda_1} \\
    &= -\frac{1-\lambda_1}{\lambda_1}\frac{1-\nu^{t}}{1-\nu} \longrightarrow -\frac{1-\lambda_1}{\lambda_1(1-\nu)}, \ \ t\rightarrow \infty,
  \end{align*}
  where we have used the inequality $\ln(1-x)\geq x/(x-1)$. Since $\ln \alpha_{t}^{-2}$ deceases and have a finite lower bound, it follows that $\alpha_{t}^{-2}$ has a positive limit. By \eqref{gamma} we have 
  \begin{align*}
    \frac{1}{\gamma_t} 
    &= \alpha_{t}^{-2}+\lambda_t\frac{1}{\gamma_{t-1}} = \alpha_{t}^{-2}+\lambda_t(\alpha_{t-1}^{-2}+\lambda_{t-1}\frac{1}{\gamma_{t-2}}) \\
    &= \cdots \\
    &= \alpha_{t}^{-2}+\lambda_{t}\alpha_{t-1}^{-2}+\lambda_{t}\lambda_{t-1}\alpha_{t-2}^{-2}+\cdots+\lambda_{t}\lambda_{t-1}\cdots\lambda_{2}\alpha_{1}^{-2}+\lambda_{t}\lambda_{t-1}\cdots\lambda_{1}\frac{1}{\gamma_{0}} \\
    &= t\alpha_{t}^{-2} + \alpha_{t}^{-2}\frac{1}{\gamma_{0}} \longrightarrow +\infty, \ \ t\rightarrow \infty
  \end{align*}
  due to $\lim_{t\rightarrow \infty}\alpha_{t}^{-2} = C_2>0$. Thus we have $\lim_{t\rightarrow \infty}\gamma_t = 0$ and
  \begin{equation*}
    \lim_{t\rightarrow \infty}\beta_t=\frac{\lim_{t\rightarrow \infty}\lambda_t}{\lim_{t\rightarrow\infty} (\lambda_t+\alpha_{t}^{-2}\gamma_{t-1})}
    = \frac{1}{1+C_2\cdot 0} = 1,
  \end{equation*}
  which are the desired results.
\end{proof}

\begin{proof}[Proof of Theorem \ref{kpu_gain}]
  Since $y=h(\boldsymbol{w},x)+\eta$ with $\eta\sim\mathcal{N}(0,\kappa^{-1}I)$, by the Bayes' formula we have
  \begin{align*}
    p(\boldsymbol{w}|y)
     &= \frac{p(y|\boldsymbol{w})p(\boldsymbol{w})}{p(y)} \\
     &\propto \exp{\left(-\frac{\kappa}{2}(y-h(\boldsymbol{w},x))^T (y-h(\boldsymbol{w},x))\right)}\exp{\left(-\frac{1}{2}(\boldsymbol{w}-\hat{\boldsymbol{w}}_t)^T\hat{P}_{t}^{-1}(\boldsymbol{w}-\hat{\boldsymbol{w}}_t)\right)} .
  \end{align*}
  Under the assumption \eqref{approx}, $\ln p(\boldsymbol{w}|y)$ is a quadratic form of $\boldsymbol{w}$, thus $\boldsymbol{w}|y$ is Gaussian. Suppose $\boldsymbol{w}|y\sim\mathcal{N}(\widehat{\boldsymbol{w}}_t, \widehat{P}_t)$. Equating the quadratic terms of $\boldsymbol{w}$ in $\ln p(\boldsymbol{w}|y)$ gives 
  \begin{equation}\label{P_gain}
    \widehat{P}_{t}^{-1} = \hat{P}_{t}^{-1}+\kappa\mathrm{D}_{\boldsymbol{w}}h(\hat{\boldsymbol{w}}_t,x)\mathrm{D}_{\boldsymbol{w}}h(\hat{\boldsymbol{w}}_t,x)^{T} .
  \end{equation}
  Therefore, by \eqref{entropy} the change of entropy of $\boldsymbol{w}$ is
  \begin{equation}\label{change}
    S_t=\frac{1}{2}\ln(\det \hat{P}_t) - \frac{1}{2}\ln(\det \widehat{P}_t)
    = \frac{1}{2}\ln\Big(\frac{\det\widehat{P}_{t}^{-1}}{\det \hat{P}_{t}^{-1}} \Big) .
  \end{equation}
  By the Weinstein--Aronszajn identity for matrix determinant, we obtain from \eqref{P_gain}
  \begin{align*}
    \det\widehat{P}_{t}^{-1} &= \det\hat{P}_{t}^{-1} \cdot \det\left(I + \kappa\hat{P}_{t}\mathrm{D}_{\boldsymbol{w}}h(\hat{\boldsymbol{w}}_t,x)\mathrm{D}_{\boldsymbol{w}}h(\hat{\boldsymbol{w}}_t,x)^{T} \right) \\
    &= \det\hat{P}_{t}^{-1} \cdot \det\left(I + \kappa\mathrm{D}_{\boldsymbol{w}}h(\hat{\boldsymbol{w}}_t,x)^{T}\hat{P}_{t}\mathrm{D}_{\boldsymbol{w}}h(\hat{\boldsymbol{w}}_t,x) \right)  \\
    &= \det \hat{P}_{t}^{-1} \cdot \det\left(I + \kappa\alpha_{t}^{-2}\mathrm{KPU}_{t}(x) \right) .
  \end{align*}
  Combining with \eqref{change} we obtain the expression of $S_t$.
\end{proof}

\begin{proof}[Proof of Theorem \ref{kpu_loss}]
  Let $\hat{y}=h(\boldsymbol{w},x)$. Then
  \begin{align*}
    \mathbb{E}_{\boldsymbol{w}}[|y-\hat{y}_t|^2] 
    &= \mathbb{E}_{\boldsymbol{w}}[|y-\hat{y}|^2] + \mathbb{E}_{\boldsymbol{w}}[|\hat{y}-\hat{y}_t|^2] + 2\mathbb{E}_{\boldsymbol{w}}[(y-\hat{y})(\hat{y}-\hat{y}_t)] \\
    &= \mathbb{E}_{\boldsymbol{w}}[|y-\hat{y}|^2] + \mathbb{E}_{\boldsymbol{w}}[|\hat{y}-\hat{y}_t|^2] + 2\eta \mathbb{E}_{\boldsymbol{w}}[h(\boldsymbol{w},x)-h(\hat{\boldsymbol{w}}_t,x)] .
  \end{align*}
  Using the approximation $h(\boldsymbol{w},x)-h(\hat{\boldsymbol{w}}_t,x)\approx\nabla_{\boldsymbol{w}}h(\hat{\boldsymbol{w}}_t,x)^T(\boldsymbol{w}-\hat{\boldsymbol{w}}_t)$ we get
  \begin{align*}
    \mathbb{E}_{\boldsymbol{w}}[h(\boldsymbol{w},x)-h(\hat{\boldsymbol{w}}_t,x)]
    &= \mathbb{E}_{\boldsymbol{w}}[\nabla_{\boldsymbol{w}}h(\hat{\boldsymbol{w}}_t,x)^T(\boldsymbol{w}-\hat{\boldsymbol{w}}_t)] \\
    &= \nabla_{\boldsymbol{w}}h(\hat{\boldsymbol{w}}_t,x)^T\mathbb{E}_{\boldsymbol{w}}[\boldsymbol{w}-\hat{\boldsymbol{w}}_t] \\
    &= 0 .
  \end{align*}
  Therefore, we have
  \begin{equation}\label{expec}
    \mathbb{E}_{\boldsymbol{w}}[|y-\hat{y}_t|^2] 
    = \mathbb{E}_{\boldsymbol{w}}[|y-\hat{y}|^2] + \mathbb{E}_{\boldsymbol{w}}[|\hat{y}-\hat{y}_t|^2] \geq \mathbb{E}_{\boldsymbol{w}}[|\hat{y}-\hat{y}_t|^2].
  \end{equation}
  Using again the approximation $h(\boldsymbol{w},x)-h(\hat{\boldsymbol{w}}_t,x)\approx\nabla_{\boldsymbol{w}}h(\hat{\boldsymbol{w}}_t,x)^T(\boldsymbol{w}-\hat{\boldsymbol{w}}_t)$
  we obtain
  \begin{align*}
    \mathbb{E}_{\boldsymbol{w}}[|\hat{y}-\hat{y}_t|^2] 
    &= \mathbb{E}_{\boldsymbol{w}}[|h(\boldsymbol{w},x)-h(\hat{\boldsymbol{w}}_t,x)|^2] \\
    &= \nabla_{\boldsymbol{w}}h(\hat{\boldsymbol{w}}_t,x)^{T}\mathbb{E}_{\boldsymbol{w}}[(\boldsymbol{w}-\hat{\boldsymbol{w}}_t)(\boldsymbol{w}-\hat{\boldsymbol{w}}_t)^{T}]\nabla_{\boldsymbol{w}}h(\hat{\boldsymbol{w}}_t,x) \\
    &= \nabla_{\boldsymbol{w}}h(\hat{\boldsymbol{w}}_t,x)^{T}\hat{P}_{t}\nabla_{\boldsymbol{w}}h(\hat{\boldsymbol{w}}_t,x) \\
    &= \alpha_{t}^{-2}\mathrm{KPU}_{t}(x) .
  \end{align*}
  Combining the above equality with \eqref{expec}, the desired result is obtained.
\end{proof}

{\small
\bibliographystyle{abbrv}  
\bibliography{ref}
}

\end{document}